\documentclass[a4paper,twocolumn,11pt,accepted=2024-04-15]{quantumarticle}
\pdfoutput=1 

\usepackage{amsmath,amsthm,amsfonts,graphicx,times,xfrac,mathtools,amssymb,bbm,verbatim,appendix,mathrsfs}
\usepackage{microtype} 	
\usepackage[british]{babel} 	

\usepackage[unicode=true,bookmarks=true,bookmarksnumbered=false,bookmarksopen=false,breaklinks=false,pdfborder={0 0 1}, backref=false,colorlinks=true]{hyperref}

\usepackage[numbers,sort&compress]{natbib}
\renewcommand*{\url}[1]{\href{#1}{#1}}

\makeatletter
\theoremstyle{plain}
\newtheorem{thm}{\protect\theoremname}
\theoremstyle{plain}
\newtheorem{lem}{\protect\lemmaname}
\theoremstyle{plain}

\theoremstyle{remark}
\newtheorem*{rem*}{\protect\remarkname}
\theoremstyle{plain}

\theoremstyle{plain}

\theoremstyle{definition}
\newtheorem{defn}{\protect\definitionname}
\theoremstyle{plain}

\theoremstyle{plain}
\newtheorem*{thm*}{\protect\theoremname}
\theoremstyle{plain}
\newtheorem*{lem*}{\protect\lemmaname}
\theoremstyle{definition}
\newtheorem{example}{\protect\examplename}

\providecommand{\propositionname}{Proposition}
\providecommand{\theoremname}{Theorem}
\providecommand{\lemmaname}{Lemma}
\providecommand{\remarkname}{Remark}
\providecommand{\conjecturename}{Conjecture}
\providecommand{\definitionname}{Definition}
\providecommand{\corollaryname}{Corollary}
\providecommand{\observationname}{Observation}
\providecommand{\examplename}{Example}
\allowdisplaybreaks

\newcommand{\trthm}{\normalfont \mathrm{tr}}

\def\bra#1{\langle{#1}\vert}
\def\ket#1{\vert{#1}\rangle}

\def\BraVert{e.g.,roup\,\mid\,\bgroup}

\def\ketbra#1#2{\vert{#1}\rangle\!\langle{#2}\vert}

\def\tr#1{\mbox{tr}\left[{#1}\right]}
\newcommand{\ptr}[2]{\mbox{tr}_{#1}\left[ #2 \right]}
\newcommand{\inp}{\normalfont \texttt{i}}
\newcommand{\out}{\normalfont \texttt{o}}

\newcommand{\Acal}{\mathcal{A}}
\newcommand{\Bcal}{\mathcal{B}}

\newcommand{\Hcal}{\mathcal{H}}

\newcommand{\Lcal}{\mathcal{L}}

\newcommand{\Kcal}{\mathcal{K}}

\begin{document}

\title{Characterising the Hierarchy of Multi-time Quantum Processes with Classical Memory}

\author{Philip Taranto}
\orcid{0000-0002-4247-3901}
\email{philipguy.taranto@phys.s.u-tokyo.ac.jp} 
\affiliation{Department of Physics, Graduate School of Science, The University of Tokyo, 7-3-1 Hongo, Bunkyo City, Tokyo 113-0033, Japan}

\author{Marco T{\'u}lio Quintino}
\orcid{0000-0003-1332-3477}
\affiliation{Sorbonne Universit{\' e}, CNRS, LIP6, F-75005 Paris, France} 

\author{Mio Murao}
\orcid{0000-0001-7861-1774}
\affiliation{Department of Physics, Graduate School of Science, The University of Tokyo, 7-3-1 Hongo, Bunkyo City, Tokyo 113-0033, Japan}

\author{Simon Milz}
\orcid{0000-0002-6987-5513}
\affiliation{School of Physics, Trinity College Dublin, Dublin 2, Ireland}
\affiliation{Trinity Quantum Alliance, Unit 16, Trinity Technology and Enterprise Centre, Pearse Street, Dublin 2, D02YN67, Ireland}

\begin{abstract}
Memory is the fundamental form of temporal complexity: when present but uncontrollable, it manifests as non-Markovian noise; conversely, if controllable, memory can be a powerful resource for information processing. Memory effects arise from/are transmitted via interactions between a system and its environment; as such, they can be either classical or quantum. From a practical standpoint, quantum processes with classical memory promise near-term applicability: they are more powerful than their memoryless counterpart, yet at the same time can be controlled over significant timeframes without being spoiled by decoherence. However, despite practical and foundational value, apart from simple two-time scenarios, the distinction between quantum and classical memory remains unexplored. Here, we analyse multi-time quantum processes with memory mechanisms that transmit only classical information forward in time. Complementing this analysis, we also study two related---but simpler to characterise---sets of processes that could also be considered to have classical memory from a structural perspective, and demonstrate that these lead to remarkably distinct phenomena in the multi-time setting. Subsequently, we systematically stratify the full hierarchy of memory effects in quantum mechanics, many levels of which collapse in the two-time setting, making our results genuinely multi-time phenomena. 
\end{abstract}

\maketitle

\section{Introduction}\label{sec::introduction}

Temporally complex memory effects appear ubiquitously across natural and engineered processes~\cite{vanKampen_2011}. Most prominently, memory can be controlled to reliably prepare states~\cite{Sayrin_2011,Magrini_2021}, simulate non-Markovian phenomena~\cite{Grimsmo_2015,Luchnikov_2019,Jorgensen_2019}, enhance information processing~\cite{Banaszek_2004,Bavaresco_2021}, and control circuit architectures~\cite{Chiribella_2008_PRL,Chiribella_2009,Mavadia_2018,White_2020}. Such primitives are routinely used in classical computers to improve performance and will be necessary for robust quantum devices~\cite{Acin_2018}. However, properly defining memory in quantum processes---a prerequisite for systematically exploiting it---has proven to be a layered question, leading to a complicated `zoo' of definitions that, in general, are incompatible with each other~\cite{Li_2018}. The main difficulty arises due to the fundamentally invasive nature of quantum measurements, which obfuscates the effects of an observer with the underlying dynamics \emph{per se} in any multi-time quantum experiment~\cite{TarantoThesis}.

Only recently has a consistent understanding of memory in quantum processes been developed via the ``process tensor'' formalism~\cite{Pollock_2018_PRL,Pollock_2018_PRA}, which has been shown to correctly generalise classical stochastic processes to the quantum realm~\cite{Milz_2020_Quantum,Strasberg_2019,Milz_2020_PRX} and encode key memory properties---length~\cite{Taranto_2019_PRL}, structure~\cite{Taranto_2019_PRA} and strength~\cite{Taranto_2021}---thereby providing a framework to analyse complex quantum dynamics ~\cite{White_2021,White_2022}. Moreover, such results have been tested experimentally~\cite{White_2020,Guo_2021}, proving their immediate applicability. Indeed, the power and flexibility of this framework is attested to by the fact that it has been repeatedly derived and applied under different guises in many contexts, the most prevalent of which being ``quantum channels with memory''~\cite{Kretschmann_2005}, ``quantum strategies''~\cite{Gutoski_2007}, ``quantum combs''~\cite{Chiribella_2008_PRL,Chiribella_2009}, ``process matrices~\cite{Oreshkov_2012}, and ``operator tensors''~\cite{Hardy_2012,Hardy_2016} (see also Refs.~\cite{Lindblad_1979,Accardi_1982,Oeckl_2003,Aharonov_2009,Cotler_2016,Portmann_2017} for related formalisms).

However, while quantum correlations in time constitute a promising resource for future quantum technologies~\cite{berk_resource_2021, berk_extracting_2021}, from a practical standpoint, control over complex \emph{quantum} memory might well be out of reach (beyond laboratory settings) in the near-term future since it requires delicate experimental setups and/or extremely short read-out times. On the other hand, manipulating quantum processes with \emph{classical} memory seems more manageable (yet still vastly more powerful than their memoryless counterparts)~\cite{Giarmatzi_2021,Nery_2021} and amenable to current quantum technologies. To date, the investigation of \emph{classical} memory in \emph{quantum} processes is still incomplete, and consequently its resourcefulness remains largely unexplored. Indeed, different experimental setups can generate distinct memory effects, and so it is crucial to distinguish between them and figure out their resourcefulness. To achieve such goals, meaningful definitions, witnesses, and physical interpretations of such processes are necessary.

Recently, an operational notion of classical memory in multi-time quantum processes has been defined~\cite{Giarmatzi_2021}; however, the full extent of its consequences in the multi-time setting remains underexplored. In the simplest scenario---the two-time setting---quantum processes with classical memory have been shown to coincide with convex mixtures of memoryless processes and are thus straightforward to analyse~\cite{Giarmatzi_2021, Nery_2021}. However, as we demonstrate, this is no longer true in the multi-time setting, where a system undergoing open dynamics (with memory) in the presence of some environment is probed at multiple times. Here, (classical) memory effects display a more intricate structure; for instance, when more than two times are considered, memory can be measured and fed forward in time to condition future dynamics. More broadly, na{\" i}vely extending intuition from the two-time setting is doomed to fail as it necessarily overlooks genuinely multi-time effects~\cite{vanKampen_1998}. 

In the multi-time setting, the question: ``what constitutes \emph{classical} memory?'' becomes ambiguous, since different, inequivalent sets of processes---each with arguably classical memory---emerge. The most natural and operationally meaningful set of quantum processes with classical memory are those that can be expressed in terms of an open quantum dynamics where the environment can only feed forward classical information. However, as we will discuss, such processes are difficult to characterise based on data that can be observed by probing the system of interest alone. Accordingly, we also consider two simpler sets of quantum processes---namely, convex mixtures of memoryless processes and those that lead to separable process tensors---which could also be considered to have ``classical memory'' (from a structural perspective) and are easier to characterise in practice, albeit at the cost of a clear physical implementation. This situation is analogous to the one considered, e.g., in entanglement theory: here, maps that can be implemented with local operations and classical communication~\textbf{(LOCC)} are a physically meaningful set, but they are difficult to characterise. Thus, for simplicity, one often considers non-equivalent but related sets of transformations, e.g., those that preserve the set of separable states~\cite{Horodecki_2009}, even though their physical implementation is unclear. 

Here, we fully characterise the naturally (due to physical or structural considerations) occurring sets of quantum processes with classical memory and prove a strict hierarchy amongst them. We demonstrate that the different types of classical memory considered can be distinguished based upon the ability of the process at hand to signal information forward in time through the environment, which---crucially---can be determined by probing the system alone. Non-signalling processes contain somewhat simple memory structures; we further show that in the multi-time setting, the relationship between classical memory, separable, and non-signalling quantum processes becomes non-hierarchical, therefore leading to a rich tapestry of potential memory effects. Our analysis complements and generalises analogous considerations in the literature for the two-time case~\cite{Giarmatzi_2021, Nery_2021}, highlighting that the multi-time scenario permits fundamentally different temporal correlations, even when only classical memory is considered. Lastly, we outline the current status regarding computational methods to characterise the sets of processes analysed.

We begin in Sec.~\ref{sec::multitimequantumprocesses} by introducing the envisaged scenario and the framework of multi-time quantum processes. Subsequently, we motivate and propose the natural definition for quantum processes with classical memory, as well as define related sets of processes that could also be considered to have classical memory in their own right (albeit with a less clear physical interpretation), stratifying the space of quantum processes with respect to the memory mechanisms. In Sec.~\ref{sec::stricthierarchy} we prove our main results, demonstrating a strict hierarchy amongst quantum processes with distinct memory structures. Following this, we analyse a number of sets of quantum processes that arise naturally in the study of memory---namely those built from classical memory, those corresponding to non-signalling processes, and those that lead to separable Choi operators---in detail, demonstrating a non-hierarchical relationship between these three sets in the multi-time setting. Lastly, in Sec.~\ref{sec::computational}, we provide methods to computationally characterise the sets of processes considered throughout this article.

\begin{figure*}[t]
\centering
\includegraphics[width=1\linewidth]{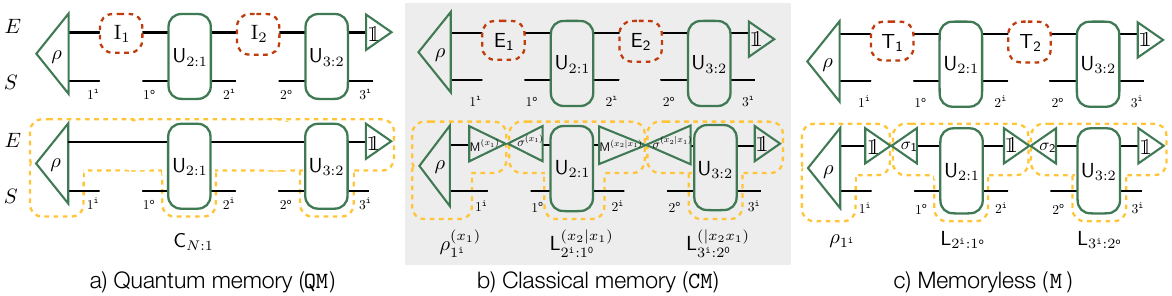}\vspace{-0.5em}
\caption{\emph{Types of Memory.} Quantum processes can have: a) quantum memory ($\texttt{QM}$, Def.~\ref{def::multitimequantumprocess}); b) classical memory ($\texttt{CM}$, Def.~\ref{def::classicalmemoryquantumprocess}); or c) no memory ($\texttt{M}$, Def.~\ref{def::memorylessquantumprocess}) (depicted for $N=3$ times). In the upper panels, we show the system-environment representation of the dynamics: in all scenarios, the system-environment evolves unitarily according to $\mathsf{U}$ between times. The distinct memory effects depend upon what happens to the environment between times: in the case of quantum memory, the environment evolves coherently, represented by the identity channel $\mathsf{I}$; in the case of classical memory, it is subject to an entanglement-breaking channel $\mathsf{E}$; in the case of no memory, it is discarded and freshly reprepared, represented by $\mathsf{T}$. Lastly, $\mathbbm{1}$ represents the tracing out over appropriate degrees of freedom. In the lower panels, we invoke the structure of the relevant environment channels to deduce the process tensor form (yellow dashed outlines). The general case cannot be broken up; the classical memory case leads to a sequence of conditional instruments as each future entanglement-breaking channel can depend upon previous outcomes [Eq.~\eqref{eq::classicalmemoryquantumprocess-2}], which is both more general than convex mixtures of memoryless processes (Def.~\ref{def::directcausequantumprocess}) and a special case of separable processes (Def.~\ref{def::separablequantumprocess}); the memoryless case leads to a sequence of independent CPTP channels [Eq.~\eqref{eq::memorylessprocess-2}].\vspace{-0.25em}}
\label{fig::typesofmemory} 
\end{figure*}

\section{Multi-time Quantum Processes}\label{sec::multitimequantumprocesses}

\subsection{Framework}\label{subsec::framework}

We consider a system $S$ that is sequentially interrogated at $N$ discrete times; in between probings, the system interacts with its environment $E$, together evolving unitarily (see Fig.~\ref{fig::typesofmemory}a). To set notation, let $\mathcal{H}$ be a finite-dimensional Hilbert space, $\mathcal{L}(\mathcal{H})$ be the set of bounded linear operators acting thereupon, and $\{\ket{i}\}_{i}$ represent the computational basis. Throughout this work, we will exclusively employ the \emph{Choi-Jamio{\l}kowski isomorphism}~\cite{Jamiolkowski_1972,Choi_1975} to represent linear maps $\mathcal{K}:\mathcal{L}(\Hcal_{\inp})\to\mathcal{L}(\Hcal_{\out})$ as matrices \mbox{$\mathsf{K}\in \Lcal(\Hcal_{\inp}\otimes\Hcal_{\out})$} via the correspondence
\begin{align}
    \mathsf{K} :=\sum_{ij} \ketbra{i}{j}\otimes \Kcal(\ketbra{i}{j}).
\end{align}
We will typically not distinguish between a linear map (in the abstract sense) and its concrete representation in terms of the Choi operator given above. However, wherever necessary to do so, we will denote the former by calligraphic script and the latter by the corresponding sans serif character. Moreover, we will find it useful to introduce the following primitives, from which we will construct complex processes. Firstly, the \emph{identity map} corresponds to
\begin{align}
    \mathrm{I}:=\sum_{ij} \ketbra{i}{j}\otimes \ketbra{i}{j},
\end{align}
and is proportional to the maximally entangled state. A \emph{unitary channel} is denoted by
\begin{align}
    \mathsf{U}:=\sum_{ij} \ketbra{i}{j}\otimes U\ketbra{i}{j}U^\dagger,
\end{align}
where $U:\Hcal_{\inp}\to\Hcal_{\out}$ is a unitary matrix. Lastly, the Choi operator of the \emph{trace map} is the identity matrix
\begin{align}
    \mathbbm{1}:=\sum_{ij} \ketbra{i}{j}\otimes \tr{\ketbra{i}{j}}=\sum_{i} \ketbra{i}{i}.
\end{align}

An object that permits the computation of all possible multi-time correlations---the ``process tensor''---can be built up in the Choi representation from these parts via the \emph{link product} $\star$~\cite{Chiribella_2009}. The link product between two Choi operators $\mathsf{A}_{12}\in \Lcal(\Hcal_1\otimes \Hcal_2)$ and $\mathsf{B}_{23}\in\Lcal(\Hcal_2\otimes \Hcal_3)$ is defined as
\begin{align}
    \mathsf{A}_{12} \star \mathsf{B}_{23} := \ptr{2}{(\mathsf{A}^{\textup{T}_2}_{12} \otimes \mathbbm{I}_3 ) \; (\mathbbm{I}_1\otimes \mathsf{B}_{23})},
\end{align}
where $\cdot^{\textup{T}_2}$ is the partial transposition on subsystem $\Hcal_2$. Intuitively, it consists of a contraction on the spaces that its constituents (here, $\mathsf{A}$ and $\mathsf{B}$) share, and a tensor product on the remaining spaces. For instance, when $\mathsf{A}$ and $\mathsf{B}$ are defined on mutually exclusive Hilbert spaces (i.e., $\Hcal_2$ is trivial), the link product reduces to the tensor product, i.e., for $\mathsf{A}_1\in \Lcal(\Hcal_1)$ and $\mathsf{B}_2\in \Lcal(\Hcal_2)$, one yields $\mathsf{A}_1\star \mathsf{B}_2=\mathsf{A}_1\otimes\mathsf{B}_2$, which is the standard way to combine independent objects in quantum theory. At the other extreme, when $\mathsf{A}$ and $\mathsf{B}$ are defined on the same Hilbert space, the link product reduces to an `inner product'-like contraction, i.e., for $\mathsf{A}\in \Lcal(\Hcal)$ and $\mathsf{B}\in \Lcal(\Hcal)$, $\mathsf{A}\star \mathsf{B}=\tr{\mathsf{A}^\textup{T}\mathsf{B}}$, which corresponds to the Born rule for appropriate objects (i.e., states and POVMs). 

More generally, when $\mathsf{A}_{12}$ and $\mathsf{B}_{23}$ represent the Choi operators of some quantum channels $\mathcal{A}:\mathcal{L}(\mathcal{H}_1)\to\mathcal{L}(\mathcal{H}_2)$ and $\mathcal{B}:\mathcal{L}(\mathcal{H}_2)\to\mathcal{L}(\mathcal{H}_3)$ respectively, the link product $\mathsf{A}_{12} \star \mathsf{B}_{23}$ corresponds to the Choi operator of the composition $\mathcal{B} \circ \mathcal{A}$. In particular, this allows one to express partial traces as, say, $\ptr{2}{\mathsf{A}_{12}} = \mathbbm{1}_{2} \star \mathsf{A}_{12}$ (since $\mathbbm{1}_{2}$ is the Choi state of the map $\text{tr}_2$) and to succinctly write the concatenation of maps as link products of Choi matrices. For example, the action of two maps $\Acal: \Lcal(\Hcal_1) \rightarrow \Lcal(\Hcal_2)$ and $\Bcal: \Lcal(\Hcal_2) \rightarrow \Lcal(\Hcal_3)$ on a state $\rho_1 \in \Lcal(\Hcal_1)$, followed by a final trace is given by $\tr{(\Bcal \circ \Acal)[\rho_1]} = \mathbbm{1}_3 \star \mathsf{B}_{23} \star \mathsf{A}_{12} \star \rho_1$. Throughout, we will exclusively employ the link product to express quantum operations and their concatenations.

An $N$-time quantum process can then be defined by concatenating global unitary dynamics with identity maps on the environment---permitting the transmission of quantum information---and a final partial trace over the environment degrees of freedom (see Fig.~\ref{fig::typesofmemory}a for an $N=3$ depiction):
\begin{defn}[Multi-time quantum process]\label{def::multitimequantumprocess}
An $N$-time \emph{quantum process} is represented by $\mathsf{C}_{N:1} \geq 0$ that can be written as
    \begin{align}\label{eq::quantummemoryprocess}
        \mathsf{C}_{N:1} &= \mathbbm{1}_{E_{N^{\inp}}} \star \Big( \bigstar_{j=1}^{N-1} \mathsf{U}_{(ES)_{j+1^{\inp} j^{\out}}}  \star \mathrm{I}_{E_{j}} \Big)\star \rho_{(ES)_{1^{\inp}}},
\end{align}
with $\rho_{(ES)_{1^{\inp}}} \in \mathcal{L}(\mathcal{H}_{E_{1^{\inp}}} \otimes \mathcal{H}_{S_{1^{\inp}}})$ the initial system-environment state; each $\mathsf{U}_{(ES)_{j+1^{\inp} j^{\out}}}$ a unitary channel [i.e., the Choi operator of the unitary map $\mathcal{U}: \mathcal{L}(\mathcal{H}_{S_{j^{\out}}}\otimes\mathcal{H}_{E_{j^{\out}}}) \to (\mathcal{H}_{S_{j+1^{\inp}}}\otimes\mathcal{H}_{E_{j+1^{\inp}}})$ corresponding to the global dynamics between times]; each $\mathrm{I}_{E_{j}} \in \mathcal{L}(\mathcal{H}_{E_{j^{\out}}}\otimes \mathcal{H}_{E_{j^{\inp}}})$ representing an identity map on the environment at each time; and
in the second line we have rewritten the final partial trace over the environment in terms of the Choi operator 
$\mathbbm{1}_{E_{N^{\inp}}} \in \mathcal{L}(\mathcal{H}_{E_{N^{\inp}}})$. We denote the set of such processes by $\mathtt{QM}$. 
\end{defn}

\noindent The Choi operator of an $N$-time process is a many-body (supernormalised) quantum state $\mathsf{C}_{N:1} \in \mathcal{L}( \mathcal{H}_{S_{N^{\inp}}} \bigotimes_{j=1}^{N-1} \mathcal{H}_{S_{j^{\out}}} \otimes \mathcal{H}_{S_{j^{\inp}}})$ defined with open slots at each time $t_1,\hdots,t_{N-1}$ and an additional final wire at $t_N$ (denoted by the subscript $N:1$). Each open slot is associated to both an input $\inp$ and output $\out$ Hilbert space, which we typically compress for conciseness. Throughout this article, we consider all processes to start on a ``slot'' and end on a final ``open line'' (see Fig.~\ref{fig::typesofmemory}). A multi-time quantum process abides by certain causality constraints that encode the impossibility of sending information from the future to the past. More precisely, a positive semidefinite operator $\mathsf{C}_{N:1} \geq 0$ represents an $N$-time process if and only if it satisfies the following \emph{causality conditions}~\cite{Chiribella_2009}
\begin{align} 
    \trthm_{N^\textup{\inp}}[\mathsf{C}_{N:1}] &= \mathbbm{1}_{{N-1}^{\out}}\otimes \mathsf{C}_{N-1:1} \nonumber \\* 
    \trthm_{N-1^\textup{\inp}}[\mathsf{C}_{N-1:1}] &=  \mathbbm{1}_{{N-2}^{\out}}\otimes\mathsf{C}_{N-2:1}\ \nonumber \\* 
  \notag  &\hspace{6pt}\vdots \nonumber \\* \label{eq::causality_conditions}
     \trthm_{1^\textup{\inp}}[\mathsf{C}_{1:1}] &= 1, 
\end{align}
where $\mathsf{C}_{k:1}\geq0$ are $k$-time quantum processes and $\mathsf{C}_{1:1}$ is a quantum state. These conditions ensure that the process \emph{cannot signal} information backwards in time. We will later use the violation of similar no-signalling conditions---albeit between arbitrary times, not just to previous ones---to distinguish between certain types of memory in quantum processes. 

Correlations of the Choi operator encode temporal correlations of the process. Recent work has demonstrated various features that arise from the interplay between memory effects and sequential measurements in quantum theory~\cite{Budini_2018, Taranto_2019_PRL, Taranto_2019_PRA, Taranto_2021, Milz_2020_PRX, Strasberg_2019, Budini_2022, Strasberg_2023, Taranto_2023}. Here, we characterise multi-time quantum processes with different types of memory, focusing on the case where the memory is restricted to being classical. 

\vspace{-0.5em}\subsection{Types of Memory in Quantum Processes}\label{subsec::typesofmemory}

We aim to delineate different types of memory, i.e., temporal correlations, similar to the hierarchies that have been developed amongst spatial correlations~\cite{Horodecki_2003, Brunner_2014}. The memory effects that can be exhibited are contingent upon what type of information the process can transmit in time. This, in turn, depends upon what happens to the \emph{environment}---the carrier of information about previous system states---during the evolution. In general, the environment evolves coherently [modelled by $\mathrm{I}_{E_{j}}$ in Eq.~\eqref{eq::quantummemoryprocess}], propagating quantum information forward. By imposing different environment dynamics, one can actively restrict the possible memory effects (see Fig.~\ref{fig::typesofmemory}).

\vspace{-0.5em}\subsubsection*{Memoryless Quantum Processes}

For instance, the environment could be erased and prepared anew in between times, thereby not transmitting \emph{any} historic information. Such \emph{refreshing} of the environment between times---which can model, e.g., rethermalisation---means that the system state at each time can be treated as a linear function of that at the immediately preceding time, which is a key assumption to deriving the much-lauded Gorini-Kossakowski-Sudarshan-Lindblad (\textbf{GKSL}) / Markovian master equations~\cite{Gorini_1976,Lindblad_1976}. Mathematically, this physical situation corresponds to a \emph{trace-and-prepare} channel being applied to the environment between times. Such channels are represented by uncorrelated Choi operators, i.e., 
\begin{align}
\mathsf{T}_{\out \inp} = \sigma_{\out} \otimes \mathbbm{1}_{\inp},
\end{align}
where $\sigma_{\out}$ is an arbitrary quantum state (which, in the picture described above, could be the thermal state of the environment). This leads to (see Fig.~\ref{fig::typesofmemory}c for an $N=3$ depiction):
\begin{defn}\label{def::memorylessquantumprocess}
    An $N$-time \emph{memoryless quantum process} is represented by a process $\mathsf{C}_{N:1}^{\textup{M}} \geq 0$ that can be written as
    \begin{align}\label{eq::memorylessprocess-1}
        \mathsf{C}_{N:1}^{\textup{M}} = \mathbbm{1}_{E_{N^{\inp}}} \star \Big(\bigstar_{j=1}^{N-1} \mathsf{U}_{(ES)_{j+1^{\inp} j^{\out}}} \star \mathsf{T}_{E_{j}}\Big) \star \rho_{(ES)_{1^{\inp}}},
    \end{align}
    with each $\mathsf{T}_{E_{j}} \in \mathcal{L}(\mathcal{H}_{E_{j^{\out}}}\otimes \mathcal{H}_{E_{j^{\inp}}})$ representing a trace-and-prepare channel on the environment at each time and the other objects as before. We denote the set of such processes by $\mathtt{M}$.
\end{defn}

\noindent By invoking the form of each trace-and-prepare channel $\mathsf{T}_{E_{j}} := \sigma_{E_{j^{\out}}} \otimes \mathbbm{1}_{E_{j^{\inp}}}$, noting that $\mathsf{L}_{j+1^{\inp}:j^{\out}} := \mathbbm{1}_{E_{j+1^{\inp}}} \star \mathsf{U}_{(ES)_{j+1^{\inp} j^{\out}}} \star \sigma_{E_{j^{\out}}}$ induces a quantum channel on the system, and defining $\rho_{1^{\inp}} := \rho_{(ES)_{1^{\inp}}} \star \mathbbm{1}_{E_{1^{\inp}}}$ (see Fig.~\ref{fig::typesofmemory}c), we have the equivalent form of a memoryless quantum process as a sequence of \emph{independent} quantum channels~\cite{Pollock_2018_PRL,Costa_2016}:
\begin{align}\label{eq::memorylessprocess-2}
    \mathsf{C}_{N:1}^{\textup{M}} = \mathsf{L}_{N^{\inp}:N-1^{\out}} \otimes \hdots \otimes \mathsf{L}_{2^{\inp}:1^{\out}} \otimes \rho_{1^{\inp}}.
\end{align}
Here, each $\mathsf{L}_{j+1^{\inp}:j^{\out}} \in \mathcal{L}(\mathcal{H}_{S_{j+1^{\inp}}}\otimes \mathcal{H}_{S_{j^{\out}}})$ represents a \emph{completely positive and trace preserving} \textbf{(CPTP)} channel evolving the system from $t_{j}$ to $t_{j+1}$; in the Choi representation, these properties are respectively reflected by $\mathsf{L}_{j+1^{\inp}:j^{\out}} \geq 0$ and $\ptr{j+1^{\inp}}{\mathsf{L}_{j+1^{\inp}:j^{\out}}} = \mathbbm{1}_{d_{j^{\out}}}$. Memorylessness of the process is clear from this uncorrelated structure, as all of the CPTP maps are mutually independent, ensuring that the environment itself plays no role in transmitting forward historic information.

\vspace{-0.5em}\subsubsection*{Classical Memory Quantum Processes}

In more general scenarios, e.g., when the environment is \emph{not} fully refreshed between times, memory can be transported. Depending on the particularities of the environment and its interactions with the system, said memory can be either classical or quantum. Here, we model classical memory effects by interpolating between the memoryless case (where the environment is refreshed between times) and the quantum memory case (where the environment coherently propagates) by interspersing the global dynamics with channels that transmit only classical information through the environment. This scenario can be modelled by subjecting the environment to an \emph{entanglement-breaking channel} \textbf{(EBC)} between times. In each run, any EBC can measure the system and feed forward classical information pertaining to the outcome. Mathematically, EBCs are described by measure-and-prepare channels, i.e., 
\begin{align}
\mathsf{E}_{\out \inp} = \sum_x \sigma_{\out}^{(x)} \otimes \mathsf{M}_{\inp}^{(x)}
\end{align}
where each $\sigma_{\out}^{(x)}$ is an arbitrary state and $\{ \mathsf{M}_{\inp}^{(x)} \}$ forms a positive operator-valued measure (\textbf{POVM}), with the label $x$ indicating a potential outcome~\cite{Horodecki_2003}. This leads to (see Fig.~\ref{fig::typesofmemory}b for an $N=3$ depiction):
\begin{defn}\label{def::classicalmemoryquantumprocess}
An $N$-time \emph{classical memory quantum process} is represented by a process $\mathsf{C}^{\textup{CM}}_{N:1} \geq 0$ that can be written as
    \begin{align}\label{eq::classicalmemoryquantumprocess-1}
        \mathsf{C}_{N:1}^{\textup{CM}} = \mathbbm{1}_{E_{N^{\inp}}} \star \Big( \bigstar_{j=1}^{N-1} \mathsf{U}_{(ES)_{j+1^{\inp} j^{\out}}} \star \mathsf{E}_{E_{j}} \Big) \star \rho_{(ES)_{1^{\inp}}},
    \end{align}
with each $\mathsf{E}_{E_{j}} \in \mathcal{L}(\mathcal{H}_{E_{j^{\out}}}\otimes \mathcal{H}_{E_{j^{\inp}}})$ representing an EBC on the environment at each time and the other objects as before. We denote the set of such processes by $\mathtt{CM}$.
\end{defn}

\noindent By explicitly invoking the structure of EBCs, we yield the form of Choi operators considered in Ref.~\cite{Giarmatzi_2021} (see App.~\ref{app::classicalmemoryquantumprocesses}):
\begin{align}\label{eq::classicalmemoryquantumprocess-2}
    \mathsf{C}_{N:1}^{\textup{CM}}\!=\!  \sum_{x_{N:1}} \mathsf{L}_{N^{\inp}:N-1^{\out}}^{(x_N|x_{N-1:1})} \otimes \hdots \otimes \mathsf{L}_{2^{\inp}:1^{\out}}^{(x_2|x_1)} \otimes \rho_{1^{\inp}}^{(x_1)},
\end{align}
where $x_{k:j} := \{ x_j,\hdots, x_k\}$, $\{ \rho_{1^{\inp}}^{(x_1)} \}$ forms a state ensemble, i.e., each $\rho_{1^{\inp}}^{(x_1)} \geq 0$ and $\rho_{1^{\inp}} := \sum_{x_1} \rho_{1^{\inp}}^{(x_1)}$ has unit trace, and \mbox{$\mathsf{L}_{j+1^{\inp}:j^{\out}}^{(x_{j+1}|x_{j:1})} := \mathsf{M}_{E_{j+1^{\inp}}}^{(x_{j+1}|x_{j:1})} \star \mathsf{U}_{(ES)_{j+1^{\inp} j^{\out}}} \star \sigma_{E_{j^{\out}}}^{(x_j|x_{j-1:1})}$} forms an instrument on the space of the system for each conditioning argument, reflecting the fact that every outcome observed can condition the choice of future EBCs (see Fig.~\ref{fig::typesofmemory}b). The intuition behind such a process is as follows. First, an EBC $\mathsf{E}_{E_1} := \sum_{x_1}\sigma^{(x_1)}_{E_{1^{\out}}} \otimes \mathsf{M}_{E_{1^{\inp}}}^{(x_1)}$ is enacted on the environment. The classical label corresponding to any observed outcome $x_1$ can be stored and fed forward to condition the overall choice of any future EBC. For instance, the second EBC can be of the form $\mathsf{E}_{E_2}^{(|x_1)} := \sum_{x_2}\sigma^{(x_2|x_1)}_{E_{2^{\out}}}\otimes \mathsf{M}_{E_{2^{\inp}}}^{(x_2|x_1)} $, which is a proper channel for each value of the conditioning argument $x_1$. In other words, for each previously recorded $x_1$, some later EBC is enacted with certainty. More generally, at any time $t_j$, the EBC to which the environment is subjected can depend upon the entire history, i.e., $\mathsf{E}_{E_j}^{(|x_{j-1:1})} := \sum_{x_j} \sigma^{(x_j|x_{j-1:1})}_{E_{j^{\out}}} \otimes \mathsf{M}_{E_{j^{\inp}}}^{(x_j|x_{j-1:1})}$, which leads to the form of \emph{conditional instruments}.

The above definition corresponds to a clear physical picture in terms of the system-environment dynamics. However, as we will discuss later, it is difficult to gauge whether a given process can be constructed by classical memory based upon observations of the system alone. As such, we will also consider two additional sets of processes that are simpler to characterise and could also be considered to have classical memory in their own right.

\vspace{-0.5em}\subsubsection*{Mixed Memoryless Quantum Processes}

The simplest such scenario is to consider classical mixtures (i.e., convex combinations) of memoryless quantum process. This set of processes appeared previously in the literature in the context of inferring causal structures under the guise of \textit{direct cause processes}~\cite{Reid_15,Feix_2017,Nery_2021}.

\begin{defn}\label{def::directcausequantumprocess}
An $N$-time \emph{mixed memoryless quantum process} is represented by a process $\mathsf{C}_{N:1}^{\textup{MM}} \geq 0$ that can be written as
\begin{align}
    \mathsf{C}_{N:1}^{\textup{MM}} = \sum_x p_x \mathsf{L}_{N^{\inp}:N-1^{\out}}^{(x)} \otimes \hdots \otimes \mathsf{L}_{2^{\inp}:1^{\out}}^{(x)} \otimes \rho_{1^{\inp}}^{(x)},
    \label{eqn::MM_proc}
\end{align}
where $p_x$ is a probability distribution and each $\mathsf{L}_{j+1^{\inp}:j^{\out}}^{(x)}\in \mathcal{L}(\mathcal{H}_{S_{j+1^{\inp}}}\otimes \mathcal{H}_{S_{j^{\out}}})$ is a CPTP channel on the system for each $x$. We denote the set of such processes by $\mathtt{MM}$.
\end{defn}

\noindent Intuitively speaking, the above definition corresponds to processes where, initially, a classical coin is flipped, and depending on the outcome $x$ (which occurs with probability $p_x$), a memoryless process is chosen. Importantly, though, such a mixed memoryless process is \emph{not} memoryless overall, as temporal correlations can be mediated via dependence on the classical label $x$ (which is, nonetheless, the only form of memory present in such a process, and therefore engenders a form of classical memory). From the outset, it seems that Def.~\ref{def::directcausequantumprocess} is a special case of Def.~\ref{def::classicalmemoryquantumprocess}; however, this is not clear for a number of reasons. Firstly, in classical probability theory, there is no distinction between processes with memory and direct cause ones, since any probability distribution $\mathbbm{P}(x_N,\hdots,x_1)$ (i.e., any classical stochastic process) can be expressed as $\sum_\lambda p_\lambda \mathbbm{P}^{(\lambda)}(x_N|x_{N-1}) \hdots \mathbbm{P}^{(\lambda)}(x_{2}|x_{1}) \mathbbm{P}^{(\lambda)}(x_1)$, where $\lambda$ acts as a hidden variable that chooses which Markovian process occurs. This follows by defining a vector $\lambda$ whose entries $\lambda_k$ take the same values as $x_k$ and setting $\mathbbm{P}^{(\lambda)}(x_k|x_{k-1}) := \delta_{\lambda_k x_k}$ and $p_{\lambda = (x_N,\hdots,x_1)} := \mathbbm{P}(x_N,\hdots,x_1)$. Thus, the classical analogues of Defs.~\ref{def::classicalmemoryquantumprocess} and~\ref{def::directcausequantumprocess} are equivalent, so it may be surprising that they differ in the quantum realm.

Secondly, if one only considers quantum processes on two times, classical memory quantum processes are equivalent to mixed memoryless quantum processes, as we demonstrate in App.~\ref{app::twotimeequivalence}.\footnote{This equivalence has been shown in Refs.~\cite{Giarmatzi_2021,Nery_2021}, but we include it here for completeness.} These two points notwithstanding, in the coming section we show that when one considers more times, direct cause classical memory quantum processes form a strict subset of classical memory ones, highlighting a genuinely multi-time phenomenon---in other words, behaviour that cannot be seen by addressing two times alone. Before doing so, however, we consider a distinct set of processes that could also bear the moniker ``classical memory''.

\vspace{-0.5em}\subsubsection*{Separable Quantum Processes}

Classical memory quantum processes have been explored in Ref.~\cite{Giarmatzi_2021}. Since mathematically characterising the set $\mathtt{CM}$ is difficult, the authors relaxed the condition that each conditional dynamics must be trace preserving, leading (not obviously) to a superset of $\mathtt{CM}$~\cite{Nery_2021}. Precisely, they dropped the demand that each $\mathsf{L}_{j^{\inp}:j-1^{\out}}^{(x_{j-1:1})} := \sum_{x_j} \mathsf{L}_{j^{\inp}:j-1^{\out}}^{(x_j|x_{j-1:1})}$ is TP [see Eq.~\eqref{eq::classicalmemoryquantumprocess-2}], only retaining positivity of each element and overall causality. This leads to:
\begin{defn}\label{def::separablequantumprocess}
An $N$-time \emph{separable quantum process} is represented by a process $\mathsf{C}_{N:1}^{\textup{SEP}} \geq 0$ that can be written as~\cite{Giarmatzi_2021,Nery_2021}
\begin{align}\label{eq::separablequantumprocess}
    \mathsf{C}_{N:1}^{\textup{SEP}} = \sum_x p_x \mathsf{L}_{N^{\inp}:N-1^{\out}}^{(x)} \otimes \hdots \otimes \mathsf{L}_{2^{\inp}:1^{\out}}^{(x)} \otimes \rho_{1^{\inp}}^{(x)},
\end{align}
where $p_x$ is a probability distribution, each \mbox{$\rho_{1^{\inp}}^{(x)} \geq 0$} and each $\mathsf{L}_{j+1^{\inp}:j^{\out}}^{(x)} \geq 0$. We denote the set of such processes by $\mathtt{SEP}$.
\end{defn}

\noindent Note the lack of constraints on the elements in this decomposition compared to Defs.~\ref{def::classicalmemoryquantumprocess} and \ref{def::directcausequantumprocess}, where they must individually represent conditional instruments or CPTP channels, respectively. Since the resulting processes are separable with respect to different points in time, their memory could also meaningfully be considered ``classical'', albeit on mathematical (i.e., the structure of $\mathsf{C}_{N:1}^{\textup{SEP}}$) rather than dynamical (i.e., the mechanism by which memory is transported) grounds. This moniker is justified since no entanglement can be shared across times, and therefore all observable temporal correlations can be explained via the (only classically correlated) separable process at hand, which plays a role analogous to a local hidden state. Separable quantum processes have been previously introduced in Refs.~\cite{Giarmatzi_2021,Nery_2021} from an abstract perspective in order to outer approximate the set of classical memory processes; for the two-time case, in Sec.~\ref{subsec::separable}, we provide a concrete circuit representation of such processes in terms of separable channels, which have been studied throughout the literature~\cite{Rains_1997,Vedral_1997,Bennett_1999}; however, in the multi-time setting, such a representation remains missing.

\emph{A priori}, it is unclear if these sets of processes (Defs.~\ref{def::classicalmemoryquantumprocess}--\ref{def::separablequantumprocess})---each representing a distinct physical and/or structural situation that could be considered to represent \textit{classical} memory---coincide on the level of what is observable on the system; we now demonstrate a clear distinction between each of them.\vspace{-0.5em}

\section{Strict Hierarchy of Memory Effects for Multi-Time Quantum Processes}\label{sec::stricthierarchy}

By imposing certain dynamics on the environment one can ``kill off'' certain types of memory---be it erasing the history (Def.~\ref{def::memorylessquantumprocess}) or destroying coherent information (Def.~\ref{def::classicalmemoryquantumprocess})---to yield distinct memory classes (note, though, that the physical picture corresponding to $\mathtt{SEP}$ is unclear). However, from a practical perspective, one cannot access the environment; thus, we now develop methods to distinguish these sets by probing the system alone, demonstrating the strict hierarchy (see Fig.~\ref{fig::hierarchy}):
\begin{thm}\label{thm::hierarchy}
For $N\geq3$ times, we have the \emph{strict} set inclusions:
    \begin{gather}
    \mathtt{M} \subsetneq \mathtt{MM} \subsetneq \mathtt{CM} \subsetneq \mathtt{SEP} \subsetneq \mathtt{QM}.
\end{gather}
\end{thm}

\begin{figure}[t]
\centering
\includegraphics[width=0.6\linewidth]{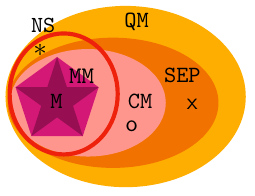}
\caption{\emph{Stratification of Memory Structures.} Here we illustrate the strict hierarchy demonstrated in Thm.~\ref{thm::hierarchy}, as well as the non-hierarchical relationship between $\mathtt{CM}$, $\mathtt{SEP}$ and $\mathtt{NS}$. In Lem.~\ref{lem::cmbutsignalling}, we present a process in $\mathtt{CM}$ but outside $\mathtt{MM}$ ($\mathsf{o}$). For two-time processes, $\mathtt{CM}$ and $\mathtt{MM}$ coincide (see App.~\ref{app::twotimeequivalence}), making our demonstration here one of a genuinely multi-time phenomenon. Furthermore, a process that is in $\mathtt{SEP}$ but outside $\mathtt{CM}$ ($\mathsf{x}$) was demonstrated in Ref.~\cite{Nery_2021}. Lastly, in Ex.~\ref{ex::entangledns} we demonstrate a process that is in $\mathtt{NS}$ but outside of $\mathtt{SEP}$ ($\mathsf{*}$).}
\label{fig::hierarchy} 
\end{figure}

\noindent Firstly, it is straightforward to note that $\mathtt{M} \subsetneq \mathtt{MM}$ due to the non-convexity of the set of memoryless processes. Similarly, $\mathtt{SEP} \subsetneq \mathtt{QM}$ holds since not all quantum processes are separable in time~\cite{Giarmatzi_2021,Milz_2021}. However, the relationship between the middle three sets---each of which positing a meaningful notion of ``classical memory quantum processes''---is more intricate. Interestingly, for processes defined on only two times, $\mathtt{MM}$ and $\mathtt{CM}$ coincide~\cite{Nery_2021} (see App.~\ref{app::twotimeequivalence}). However, on more than two times, the sets differ. Demonstrating this gap requires constructing a classical memory quantum process, i.e., one of the form given in Def.~\ref{def::classicalmemoryquantumprocess}, that is not a mixed memory quantum process, i.e., \emph{cannot} be rewritten in the form of Def.~\ref{def::directcausequantumprocess}. Analogously to the case of finding whether or not a given state has a separable decomposition, this task is difficult in general; even more so here as both sets that we wish to distinguish correspond to separable Choi states, with the difference being a constraint on each element of the decomposition individually (i.e., the additional trace preservation of each channel in the decomposition required in the latter case). We now move to develop a simple criteria that can serve as a \emph{witness} to the impossibility of a mixed memory realisation, thereby allowing us to prove said distinction.

\begin{table*}
    \centering
    \begin{tabular}{lll}
        \textbf{Name} & \textbf{Definition} & \textbf{Constraint} \\ \hline\hline
        Memoryless $\texttt{M}$ & Def.~\ref{def::memorylessquantumprocess} & Tensor product of CPTP channels [see Eq.~\eqref{eq::memorylessprocess-2}] \\ \hline
        Mixed Memory $\texttt{MM}$ & Def.~\ref{def::directcausequantumprocess} & Convex combination of memoryless processes [see Eq.~\eqref{eqn::MM_proc}] \\ \hline
        Classical Memory $\texttt{CM}$ & Def.~\ref{def::classicalmemoryquantumprocess} & Composition of conditional instruments [see Eq.~\eqref{eq::classicalmemoryquantumprocess-2}] \\ \hline
        Separable $\texttt{SEP}$ & Def.~\ref{def::separablequantumprocess} & Separable Choi operator [see Eq.~\eqref{eq::separablequantumprocess}] \\ \hline
        Non-Signalling $\texttt{NS}$ & Def.~\ref{def::nsquantumprocess} & No signalling forwards in time [see Eq.~\eqref{eq::nsquantumprocess}] \\ \hline
        Quantum Memory $\texttt{QM}$ & Def.~\ref{def::multitimequantumprocess} & General quantum process [see Eq.~\eqref{eq::causality_conditions}] 
    \end{tabular}
\caption{\emph{Classes of Quantum Processes.} Here, we summarise the six different types of quantum processes that we consider throughout this article. The formal definitions are provided in the middle column, with the rightmost column conveying an intuitive description of the relevant constraints on the Choi state of the process. Furthermore, the underlying physical realisation of processes in the sets $\mathsf{QM}$, $\mathsf{CM}$, and $\mathsf{M}$ are depicted in Fig.~\ref{fig::typesofmemory}; processes in the set $\mathsf{MM}$ can be implemented by randomly sampling memoryless processes; and finally, processes in the sets $\mathsf{SEP}$ and $\mathsf{NS}$ do not have a clear underlying physical realisation, but admit an explicit characterisation via their Choi states.}
    \label{tab::processes}
\end{table*}

As it turns out, many of the steps in the proof of Thm.~\ref{thm::hierarchy} can be simplified by analysing the \emph{signalling} properties of the considered processes. This motivates the definition of the set of non-signalling ($\mathtt{NS}$) processes: 
\begin{defn}
\label{def::nsquantumprocess}
An $N$-time \emph{non-signalling quantum process} is represented by a process $\mathsf{C}_{N:1}^{\textup{NS}} \geq 0$ that satisfies
\begin{align}\label{eq::nsquantumprocess}
        \trthm_{k^\textup{\inp}}[\mathsf{C}_{N:1}^{\textup{NS}}] = \frac{\mathbbm{1}_{{k-1}^{\out}}}{d_{{k-1}^{\out}}} \otimes \trthm_{k^\textup{\inp}{k-1}^\textup{\out}}\!\left[\mathsf{C}_{N:1}^{\textup{NS}}\right]\, ,
\end{align}
$\forall \, k$. We denote the set of such processes by $\mathtt{NS}$.
\end{defn}

\noindent Having introduced all of the classes of processes relevant for our analysis, we provide a summary in Tab.~\ref{tab::processes} to aid in understanding their differences. Besides being an interesting set of processes in their own right, $\mathtt{NS}$ processes provide a simple and intuitive way to understand the strict inclusions of Thm.~\ref{thm::hierarchy}. Consequently, we  will use this set of non-signalling processes as an auxiliary set that aids in the proof of Thm.~\ref{thm::hierarchy}, as depicted in Fig.~\ref{fig::proofrelations}. Importantly, such processes do not allow for signalling between non-adjacent times by preparing different states at any time $t_j$. More concretely, as soon as the system is discarded at a time $t_k$---corresponding to the operation $\trthm_{k^\textup{\inp}}$ on the l.h.s. of Eq.~\eqref{eq::nsquantumprocess}---the potential influence of any state that was fed in at the preceding time $t_{k-1}$ is discarded as well---corresponding to the identity matrix on space associated to time $k-1^{\out}$ on the r.h.s. of Eq.~\eqref{eq::nsquantumprocess}. As a consequence, in such a process it is impossible to send signals forward beyond the subsequent time by preparing different system states. 

This property is analogous to the special case of multi-partite channels where any of the involved parties $A, B, C, ...$ can \emph{only} send signals to its respective output $A', B', C', ...$ but not to any other party (see Sec.~\ref{subsec::nonsignalling} for a more detailed discussion of the non-signalling property). Throughout the literature, such processes have been discussed under the guise of causal/non-signalling quantum \emph{channels} $\Lambda_\mathtt{NS}: \Lcal(\bigotimes_{j=1}^{N-1} \mathcal{H}_{S_{j^{\out}}}) \rightarrow \Lcal(\bigotimes_{k=1}^{N} \mathcal{H}_{S_{k^{\inp}}})$~\cite{beckman_causal_2001, piani_properties_2006}. Here, instead, we consider them as multi-time quantum processes, which does not change their mathematical representation (in terms of Choi operators), but makes the interpretation of their signalling properties somewhat more subtle (see discussion in Sec.~\ref{subsec::nonsignalling}). We emphasise that every non-signalling quantum channel $\Lambda_\mathtt{NS}$ can indeed be understood as a multi-time quantum process, since---as can be seen by direct comparison---satisfaction of the process non-signalling constraints~\eqref{eq::nsquantumprocess} (which correspond to appropriate similar conditions for any non-signalling channel) automatically implies that the causality conditions of Eq.~\eqref{eq::causality_conditions} are obeyed.

Both memoryless and $\mathtt{MM}$ processes lie within $\mathtt{NS}$ (see Lem.~\ref{lem::MMnosignalling}). This is trivially true for $N=2$ times, in which case we have $\mathtt{NS} = \mathtt{QM}$ and the strict hierarchy of Thm.~\ref{thm::hierarchy} breaks down (see App.~\ref{app::twotimeequivalence}). On the other hand, for $N\geq 3$, the situation presents itself more differentiated and we have strict inclusions of the different sets of memory processes that can be proven via their relation to $\mathtt{NS}$. To this end, we first have:

\begin{lem}\label{lem::DCsubsetofCM}
For $N\geq3$ times, mixed memoryless quantum processes are a strict subset of classical memory quantum processes, \textup{i.e.}, \mbox{$\mathtt{MM} \subsetneq \mathtt{CM}$}.
\end{lem}

\noindent The full proof of this claim is presented in App.~\ref{app::cmnotMM}. We first derive a necessary condition that must be satisfied by all $\mathtt{MM}$ processes, namely that they can only signal from any output time to the next input. More precisely, we make use of:

\begin{figure}[t]
\centering
\includegraphics[width=0.85\linewidth]{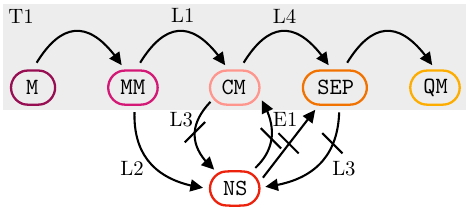}\vspace{-0.375em}
\caption{\emph{Schematic of the Proven Implications Relevant for Thm.~\ref{thm::hierarchy} (T1).} Arrows denote \emph{strict} set inclusion; arrows with a slash denote non-inclusion. Lem.~\ref{lem::DCsubsetofCM} (L1) results from Lem.~\ref{lem::MMnosignalling} (L2) and \ref{lem::cmbutsignalling} (L3). Lem.~\ref{lem::nery} (L4) shows that $\mathtt{CM}$ is included in $\mathtt{SEP}$; along with L3, this implies that $\mathtt{SEP}$ is not included within $\mathtt{NS}$. Finally, the fact that $\mathtt{NS}$ processes can be entangled [Ex.~\ref{ex::entangledns} (E1)] implies that it is not included in $\mathtt{SEP}$ (or, therefore, $\mathtt{CM})$; hence, there is no hierarchical relationship amongst $\mathtt{CM}$, $\mathtt{SEP}$, and $\mathtt{NS}$. }
\label{fig::proofrelations} 
\end{figure}

\begin{lem}\label{lem::MMnosignalling}
For $N\geq3$ times, mixed memoryless quantum processes are a strict subset of non-signalling quantum processes, \textup{i.e.}, $\mathtt{MM} \subsetneq \mathtt{NS}$.
\end{lem}

\noindent In contradistinction, processes in $\mathtt{CM}$ can signal arbitrarily far into the future by feeding forward classical information. Thus, it is clear that mixed memoryless quantum processes are subject to more constraints than classical memory ones; however, it is \emph{a priori} unclear if there exist classical memory processes that remain outside of the set of mixed memoryless ones. To prove this statement, we present an explicit example of a $\mathtt{CM}$ process that violates the aforementioned no-signalling condition, thereby ultimately precluding a $\mathtt{MM}$ explanation. Formally, we present:\begin{lem}\label{lem::cmbutsignalling}
For $N\geq3$ times, there exist classical memory quantum processes that are outside of the set of non-signalling quantum processes, \textup{i.e.}, $\mathtt{CM} \not\subseteq \mathtt{NS}$.
\end{lem}

\noindent The process, depicted in Fig.~\ref{fig::circuit}, swaps whatever state is input after $t_1$ with a fiducial environment state, then performs an EBC on the environment (ensuring implementability with only classical memory) followed by enacting a $\mathtt{CNOT}$ gate on the system between times $t_2$ and $t_3$. It is straightforward to see that this process can signal from $t_1$ to $t_3$ (by preparing two different states $\tau_1$ and $\tau_1^\prime$ and noting that the final output state varies as a function of any state $\sigma_2$ input at time $t_2$). Invoking Lem.~\ref{lem::MMnosignalling}, such signalling is incompatible with any $\mathtt{MM}$ realisation, although the process is in $\mathtt{CM}$ by construction; hence we have proven Lem.~\ref{lem::DCsubsetofCM}.  

\begin{figure}[t]
\centering
\includegraphics[width=0.75\linewidth]{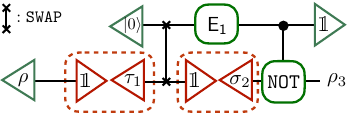}\vspace{-0.5em}
\caption{\emph{A process in $\mathtt{CM}$ but not $\mathtt{NS}$.} The process (green) above in $\mathtt{CM}$ by construction and can signal from $t_1$ to $t_3$; hence it proves Lem.~\ref{lem::cmbutsignalling}. This can be explicitly tested by tracing out the system at times $t_1$ and $t_2$, fixing a state $\sigma_2$, and feeding in an arbitrary state $\tau_1$ (red). For instance, $\rho_3(\tau_1 = \ket{0}) = \sigma_2$ whilst \mbox{$\rho_3(\tau_1^\prime=\ket{1}) = \mathtt{NOT}(\sigma_2)$}. }
\label{fig::circuit} 
\end{figure}

Finally, it is easy to see that $\mathtt{CM} \subseteq \mathtt{SEP}$, since all $\mathtt{CM}$ processes are separable in the Choi representation [see Eq.~\eqref{eq::classicalmemoryquantumprocess-2}]. The \emph{strict} inclusion $\mathtt{CM} \subsetneq \mathtt{SEP}$ was conjectured in Ref.~\cite{Giarmatzi_2021} and proven in Ref.~\cite{Nery_2021}. This distinction is akin to the existence of maps that preserve separable states but cannot be implemented by local operations and classical communication~\cite{Horodecki_2009}; although the former are easier to characterise, the latter are more operationally meaningful. Similarly, since the constraints on separable multi-time processes $\mathtt{SEP}$ are precisely those of multi-partite entanglement of quantum states, one can employ the entire machinery from entanglement theory, e.g., negativity under partial transposition~\cite{Peres_1996,Horodecki_1996}, entanglement witnesses~\cite{Guhne_2009}, and the k-symmetric extension hierarchy~\cite{Doherty_2004}. However, given a process in $\mathtt{SEP}$, there does not generally exist a realisation using only classical memory~\cite{Nery_2021}, as formalised in:
\begin{lem}[\cite{Nery_2021}]\label{lem::nery}
For $N\geq2$ times, there exist separable quantum processes that are outside of the set of classical memory quantum processes, \textup{i.e.}, $\mathtt{CM} \subsetneq \mathtt{SEP}$.
\end{lem}

\noindent Note that the above lemma holds for $N\geq2$ times (in contrast to the previous statements which require $N\geq3$). The proof presented in Ref.~\cite{Nery_2021} shows that the following separable $N=2$-time process
\begin{align} \label{eq::guerin}
\mathsf{C}^{\textup{SEP}}_{2:1} = \frac{1}{2} &\left( \ket{000}\bra{000} + \ket{101}\bra{101} \right. + \notag \\* 
    &\left. \ket{+1+}\bra{+1+} + \ket{-1-}\bra{-1-} \right),
\end{align}
where $\mathsf{C}_{2:1}^{\mathrm{SEP}} \in \Lcal(\Hcal_{2^{\inp}} \otimes \Hcal_{1^{\out}} \otimes \Hcal_{1^{\inp}})$, does \emph{not} admit a classical memory realisation. Notice that, one can easily  ``bootstrap'' this example to $N>2$-time scenarios by simply considering that the other parties to be one-dimensional, or by tensoring $\mathsf{C}^{\textup{SEP}}_{2:1}$ with an arbitrary process on the other times. Together, Lems.~\ref{lem::nery} and~\ref{lem::DCsubsetofCM} thus complete the proof of Thm.~\ref{thm::hierarchy}.

\vspace{-0.5em}\subsection{Non-signalling Quantum Processes \mbox{Revisited}}\label{subsec::nonsignalling}

We now move to analyse the set $\mathtt{NS}$ of non-signalling processes that we introduced above, which proved to be a helpful auxiliary set for the proof of Thm.~\ref{thm::hierarchy}. These processes constitute a rich class in their own right and can exhibit many interesting phenomena, including genuinely quantum correlations (see Ex.~\ref{ex::entangledns}) that nonetheless do not allow for signalling between arbitrary (non-adjacent) times $t_j$ and $t_k$ $(t_k>t_j)$ by preparing different quantum states at $t_j$. This can be seen by means of an explicit example:
\begin{example}\label{ex::entangledns}
    Consider a process in which an $N$-partite state $\rho_{N^{\inp}\cdots 1^{\inp}}$ is initially prepared, with each subsystem being fed out at each time and the process simply discarding the post-measurement states. Any such process is represented by
    \begin{align}
        \widetilde{\mathsf{C}}_{N:1}^\textup{NS} = \rho_{N^{\inp}\cdots 1^{\inp}} \otimes \mathbbm{1}_{{N-1}^{\out} \cdots 1^{\out}},
    \end{align}
    Clearly, $\widetilde{\mathsf{C}}_{N:1}^\textup{NS}$ is non-signalling. On the other hand, it can exhibit genuinely quantum correlations between arbitrary times; in particular, whenever $\rho_{N^{\inp}\cdots 1^{\inp}}$ is entangled, $\widetilde{\mathsf{C}}_{N:1}^\textup{NS} \not\in \mathtt{SEP}$.
\end{example}

\noindent This example coincides with the definition of \emph{quantum common cause} processes that have been studied throughout the literature~\cite{Costa_2016,Feix_2017,Guo_2021,Nery_2021}, where the initial state that is fed out over times is considered the cause of all observed correlations, rather than any active interventions. Such behaviour is reminiscent of non-signalling channels, where each input can \textit{only} signal to its corresponding output, but not to any other one, with this property holding true even when allowing for correlated input states~\cite{piani_properties_2006}. 

When considered as multi-time processes, however, the notion of signalling presents itself more subtly, since each timeslot contains both an input \emph{and} an output, making quantum memory---and thus signalling properties---inherently instrument-dependent~\cite{Taranto_2019_PRL, Taranto_2019_PRA, Taranto_2021, Taranto_2023}. Specifically, the non-signalling property of $\mathtt{NS}$ processes states that any state preparation at time $t_j$ (with corresponding Choi operator $\mathsf{T}_{j} = \sigma_{j^{\out}} \otimes \mathbbm{1}_{j^{\inp}}$)  blocks all signalling from earlier times $t_i < t_j$ to future times $t_k > t_j$. However, this does \textit{not} imply that there is no possibility of signalling between non-adjacent times in $\mathtt{NS}$ processes when \emph{different} operations are performed; additionally, such processes can display memory effects of arbitrary length. To see the former point, one can, e.g., consider a memoryless process that only consists of an initial state preparation and trivial dynamics between times, i.e., 
\begin{align}
    \widehat{\mathsf{C}}^{\textup{NS}}_{N:1} = \mathrm{I}_{N^{\inp}{N-1}^{\out}} \otimes \cdots \otimes \mathrm{I}_{2^{\inp}{1}^{\out}} \otimes \rho_{1^{\inp}}.
\end{align}
Now, performing two different CPTP maps $\mathsf{N}^{(\alpha)}_{j^{\out}j^{\inp}}$ (with $\alpha \in \{0,1\}$) at time $t_j$ (and doing nothing the remaining times) yields the state $\rho^{(\alpha)} = \mathcal{N}^{(\alpha)}[\rho]$ at any time $t_k > t_j$, i.e., the choice $\alpha$ of which CPTP map was performed at $t_j$ can be transmitted in forward in time (arbitrarily far), despite $\mathsf{C}^{\textup{NS}}_{N:1}$ obviously being a non-signalling process. 

Nevertheless, the process $\widehat{\mathsf{C}}^{\textup{NS}}_{N:1}$ described above is still memoryless: when only \emph{causal breaks}---i.e., instruments that exclusively contain independent measure-and-prepare operations $\mathsf{B}_{j} = \sigma_{j^{\out}} \otimes \mathsf{M}_{j^{\inp}}$---are permitted at time $t_j$, then no signalling between non-adjacent times can be detected~\cite{Pollock_2018_PRA, Pollock_2018_PRL}, even on the level of individually recorded measurement outcomes. However, as is evident from Ex.~\ref{ex::entangledns}, $\mathtt{NS}$ processes can admit (quantum) memory effects between arbitrarily spaced times $t_j$ and $t_k$, even when only causal breaks are permitted. There, the $\mathtt{NS}$ process $\widetilde{\mathsf{C}}_{N:1}^\textup{NS} = \rho_{N^{\inp}\cdots 1^{\inp}} \otimes \mathbbm{1}_{{N-1}^{\out} \cdots 1^{\out}}$ was introduced. Evidently, the correlations exhibited by this processes reflect those that are present in $\rho_{N^{\inp}\cdots 1^{\inp}}$, and hence can concern arbitrary times and be fully detectable via causal breaks~\cite{Taranto_2019_PRA} (such processes then even allow for genuine multi-partite entanglement between different times, a phenomenon that is also present for two-time $\mathtt{NS}$ processes~\cite{Milz_2021}). 

Finally, we remark that the fact that all trace-and-prepare operations $\mathsf{T}_j = \sigma_{j^{\out}}\otimes \mathbbm{1}_{j^{\out}}$ block any signalling in $\mathtt{NS}$ processes may be viewed as \emph{memorylessness on average}. A process is memoryless if there do not exist two different causal breaks $\mathsf{B}_j^{(\alpha)} = \sigma_{j^{\out}} \otimes \mathsf{M}^{(\alpha)}_{j^{\inp}}$ with the same state fed forward but different POVM elements $\mathsf{M}^{(\alpha=1,2)}_{j^{\inp}}$, such that information from the past can be detected at times after $t_j$~\cite{Pollock_2018_PRA, Pollock_2018_PRL}. In general, $\mathtt{NS}$ processes do \emph{not} satisfy this property. However, \textit{averaging} over all possible outcomes $\alpha$ in any causal break yields a trace-and-prepare operation $\mathsf{T}_j = \sum_\alpha \mathsf{B}^{(\alpha)}_j = \sigma_{j^{\out}} \otimes \mathbbm{1}_{j^{\inp}}$, which blocks \emph{all} influence of the past on the future for any non-signalling process. This, in turn, provides us with the intuitive interpretation of $\mathtt{NS}$ processes as those processes that are memoryless on average, but not necessarily at a fine-grained level.

\vspace{-0.5em}\subsection{Separable Quantum Processes Revisited}\label{subsec::separable}

Analogous to $\mathtt{NS}$ quantum processes---and in stark contrast to both $\mathtt{MM}$ and $\mathtt{CM}$ quantum processes---beyond the two-time case it is not straightforward to provide a circuit representation of processes that lie in $\mathtt{SEP}$. To provide this characterisation in the two-time case, it is helpful to first define the set of separable \emph{channels}~\cite{Rains_1997, vedral_entanglement_1998, Bennett_1999}: 
\begin{defn}
    \emph{Separable channels} $\Lambda^{\textup{SC}} : \quad \Lcal(\Hcal_{A} \otimes \Hcal_{B}) \rightarrow  \Lcal(\Hcal_{A'} \otimes \Hcal_{B'})$ are quantum channels that can be written as \mbox{$\Lambda^{\textup{SC}}[X] = \sum_{k} (\alpha_k \otimes \beta_k)X(\alpha_k^\dagger \otimes \beta_k^\dagger)$}, where $\alpha_k: \Hcal_A \rightarrow \Hcal_{A'}$ and $\beta_k: \Hcal_B \rightarrow \Hcal_{B'}$ for all $k$. We denote the set of separable channels by $\mathtt{SC}$.
\end{defn}

\noindent Although structurally related, we emphasise that the elements of $\mathtt{SEP}$ and $\mathtt{SC}$ do not coincide (and are even different types of objects). However, for the two-time case, they are intimately connected: here, we show that any separable process of the form 
\begin{gather}
    \mathsf{C}_{2:1}^{\mathrm{SEP}} = \sum_x p_x \mathsf{L}_{2^{\inp}1^{\out}}^{(x)} \otimes \rho_{1^{\inp}}^{(x)},
\end{gather}
where $\mathsf{C}_{2:1}^{\mathrm{SEP}} \in \Lcal(\Hcal_{2^{\inp}} \otimes \Hcal_{1^{\out}} \otimes \Hcal_{1^{\inp}})$,
can be represented as a quantum process engendered by a circuit with an initial pure quantum state $\xi_{1^{\inp}E_{1}}$ and a separable channel $\Lambda^{\textup{SC}}: \Lcal(\Hcal_{1^{\out}} \otimes \Hcal_{E_{1}}) \rightarrow \Lcal(\Hcal_{2^{\inp}})$, where $\Lambda^{\textup{SC}}$ has Kraus operators $\alpha_k: \Hcal_{E_1} \rightarrow \mathbbm{C}$ and $\beta_k: \Hcal_{1^{\out}} \rightarrow  \Hcal_{2^{\inp}}$. This can be seen directly by considering the canonical way of constructing circuit representations of quantum processes~\cite{Chiribella_2009, Bisio_2011}. 

Concretely, let $\rho_{1^{\inp}} := \tfrac{1}{d_{1^{\out}}} \trthm_{2^{\inp}1^{\out}}[\mathsf{C}^{\mathrm{SEP}}_{2:1}]$ be the initial state of the system, and  $\xi_{1^{\inp}E} := \ketbra{\xi_{1^{\inp}E}}{\xi_{1^{\inp}E}} = \rho_{1^{\inp}}^{\frac{1}{2}} \ketbra{\Phi^+_{1^{\inp}E}}{\Phi^+_{1^{\inp}E}} \rho_{1^{\inp}}^{\frac{1}{2}}$ be its canonical dilation, where we set $\Hcal_{E} \cong \Hcal_{1^{\inp}}$ and $\ket{\Phi^+_{1^{\inp}E}} \in \mathcal{H}_{1^{\inp}} \otimes \mathcal{H}_E$ is the unnormalised maximally entangled state. Now, let $\eta_{2^{\inp} 1^{\out} 1^{\inp}} := \rho_{1^{\inp}}^{-\frac{1}{2}} \mathsf{C}_{2:1}^{\mathrm{SEP}} \rho_{1^{\inp}}^{-\frac{1}{2}}$, where $\rho_{1^{\inp}}^{-\frac{1}{2}}$ is the Moore-Penrose pseudoinverse of $\rho_{1^{\inp}}^{\frac{1}{2}}$. For notational convenience, we now relabel $1^{\inp} \leftrightarrow E$ in $\eta_{2^{\inp} 1^{\out} 1^{\inp}}$, such that $\eta_{2^{\inp} 1^{\out}  E} \in \Lcal(\Hcal_{2^{\inp}} \otimes \Hcal_{1^{\out}} \otimes \Hcal_{E} )$. As a consequence of the causality constraints on $\mathsf{C}^{\mathrm{SEP}}_{2:1}$, we have that $\trthm_{2^{\inp}}[\eta_{2^{\inp} 1^{\out} E}] = \mathbbm{1}_{1^{\out} E}$, i.e., it is the Choi operator of a CPTP map $\Lambda^{\textup{SC}}: \Lcal(\Hcal_{E}\otimes \Hcal_{1^{\out}}) \rightarrow \Lcal(\Hcal_{2^{\inp}})$. Now, it is easy to see that $\Lambda^{\textup{SC}} \in \mathtt{SC}$, since its Choi operator $\eta_{2^{\inp} 1^{\out} E} = \sum_{x} p_x \mathsf{L}_{2^{\inp}1^{\out}}^{(x)} \otimes \widetilde{\rho}_{E}^{\,(x)}$, where $\widetilde{\rho}_{E}^{\,(x)} := \rho_{E}^{-\frac{1}{2}} \rho_{E}^{(x)}$, is separable (with respect to the $2^{\inp} 1^{\out} | E$ cut), which implies that its Kraus operators are of the form $\alpha_k \otimes \beta_k$. Additionally, it is straightforward to verify that
\begin{gather}
 \mathsf{C}_{2:1}^{\mathrm{SEP}} = \xi_{1^{\inp}E} \star \eta_{2^{\out} 1^{\out} E}\, 
\end{gather}
holds, which shows that any two-time process in $\mathtt{SEP}$ can be obtained from a separable channel in $\mathtt{SC}$. 

We emphasise that this correspondence does not necessarily hold true in the multi-time setting. There, although it is obvious that any quantum process consisting of only channels from $\mathtt{SC}$ is an element of $\mathtt{SEP}$, it is unclear if the converse implication holds. 

\section{Computational Characterisations}\label{sec::computational}

We now discuss the problem of deciding whether a given process is an element of each of the sets discussed in this work. Before proceeding, we first remark that checking whether a linear operator is positive semidefinite can be achieved by first checking if it is self-adjoint and then verifying that all its eigenvalues are non-negative. Alternatively, the Cholesky decomposition provides an efficient method to ensure that a given matrix is positive semidefinite.

\vspace{-0.75em}\subsubsection*{\texorpdfstring{$\mathtt{QM}:$ }{} SDP Characterisation}

In order to check if a positive semidefinite operator \mbox{$\mathsf{C}_{N:1}\geq0$} is a valid quantum memory process (i.e., belongs to $\mathtt{QM}$), one must verify the causality conditions presented in Eq.~\eqref{eq::causality_conditions}, which can be evaluated explicitly. These conditions must be checked first to ensure that the operator represents a valid process, before one tries to determine which of the set of processes it fits within (see below). Since all such conditions correspond to linear and positive semidefinite constraints on the Choi operator, the set $\mathtt{QM}$ admits characterisation means of a \emph{semidefinite program} \textbf{(SDP)}~\cite{Chiribella_2016,Skrzypczyk_2023}, which can be solved efficiently. 

\vspace{-0.75em}\subsubsection*{\texorpdfstring{$\mathtt{M}:$}{} Simple Characterisation}

Unlike all other sets of processes we consider, the set of memoryless process is not convex. Nonetheless, checking whether a quantum process $\mathsf{C}_{N:1}\in\mathtt{QM}$ represents a memoryless one (i.e., belongs to $\mathtt{M}$) is straightforward. One can evaluate $\widetilde{\mathsf{L}}_{k^{\inp} k-1^{\out}} := \ptr{\overline{k^{\inp} k-1^{\out}}}{\mathsf{C}_{N:1}}/d_{\overline{k-1^{\out}}}$ where $\overline{X}$ denotes the complement of $X$ and $d_{{\overline{k-1^{\out}}}}:=\tr{\mathbbm{1}_{{\overline{k-1^{\out}}}}}$, then construct 
\begin{align} \label{eq::memorylessprocess-3}
    \widetilde{\mathsf{C}}_{N:1}= \widetilde{\mathsf{L}}_{N^{\inp}:N-1^{\out}} \otimes \hdots \otimes \widetilde{\mathsf{L}}_{2^{\inp}:1^{\out}} \otimes \widetilde{\rho}_{1^{\inp}},
\end{align}
and check if $\widetilde{\mathsf{C}}_{N:1} = \mathsf{C}_{N:1}$. If this equation is false, then $\mathsf{C}_{N:1}$ is not a memoryless quantum process. 

Additionally, one can quantify the amount of memory $\mathcal{N}$ in a process by measuring the distance $\mathcal{D}$ to its nearest memoryless counterpart~\cite{Pollock_2018_PRL}, i.e.,
\begin{align}
    \mathcal{N}(\mathsf{C}_{N:1}) := \min_{\mathsf{C}_{N:1}^{\textup{M}} \in \mathtt{M}}\mathcal{D}(\mathsf{C}_{N:1}\|\mathsf{C}_{N:1}^{\textup{M}}).
\end{align}
Choosing the quantum relative entropy $\mathcal{S}$ as an appropriate (pseudo-)distance avoids the need for performing the above minimisation, since the nearest memoryless process is guaranteed to be the one constructed from the product of its marginals [as per Eq.~\eqref{eq::memorylessprocess-3}]. This can be seen as follows: consider the difference \mbox{$\mathcal{S}(\rho_{AB}\|\tau_A\otimes \tau_B) - \mathcal{S}(\rho_{AB}\|\rho_A\otimes \rho_B)$}, where $\tau_{A/B}$ are generic states and $\rho_{A/B} := \ptr{B/A}{\rho_{AB}}$ are the marginals of the global state. By definition of the quantum relative entropy $\mathcal{S}(A\|B) := \tr{A \log(A-B)}$ and using $\log(A\otimes B) = \log(A) \otimes \mathbbm{1}_B + \mathbbm{1}_A \otimes \log(B)$, this quantity reduces to $\mathcal{S}(\rho_A \| \tau_A) + \mathcal{S}(\rho_B \| \tau_B)$. Minimising the original quantity over uncorrelated states then amounts to minimising this last sum. Finally noting that the quantum relative entropy is non-negative and vanishes iff the arguments are equal asserts the claim. Hence, the above formula reduces to
\begin{align}
    \mathcal{N}(\mathsf{C}_{N:1}) = \mathcal{S}(\mathsf{C}_{N:1}\|\widetilde{\mathsf{C}}_{N:1}).
\end{align}
A process is memoryless iff this quantity vanishes.

\vspace{-0.75em}\subsubsection*{\texorpdfstring{$\mathtt{NS}:$ }{} SDP Characterisation}

In order to check if a quantum process $\mathsf{C}_{N:1}\in\mathtt{QM}$ is a non-signalling one (i.e., belongs to $\mathtt{NS}$), one must verify the non-signalling conditions in Eq.~\eqref{eq::nsquantumprocess}. Since all such conditions are linear and positive semidefinite constraints, the set of non-signalling processes $\mathtt{NS}$ admits an SDP characterisation.

\vspace{-0.75em}\subsubsection*{\texorpdfstring{$\mathtt{SEP}:$ }{} SDP Hierarchy Characterisation}

To check if a quantum process $\mathsf{C}_{N:1}\in\mathtt{QM}$ is a separable one (i.e., belongs to $\mathtt{SEP}$), one must ensure that the (multi-partite) operator $\mathsf{C}_{N:1}$ is separable---a problem known to be difficult~\cite{Gurvits_2003}. One may employ methods from entanglement theory, such as negativity under partial transposition~\cite{Peres_1996,Horodecki_1996}, entanglement witnesses~\cite{Guhne_2009}, and the k-symmetric extension hierarchy~\cite{Doherty_2004}. The last technique provides a characterisation in terms of a converging hierarchy of SDPs.

\vspace{-0.75em}\subsubsection*{\texorpdfstring{$\mathtt{MM}:$ }{} Polytope / SDP Hierarchy \mbox{Characterisation}}

Checking if a quantum process $\mathsf{C}_{N:1}\in\mathtt{QM}$ is a mixed memoryless one (i.e., belongs to $\mathtt{MM}$) is similar to checking if it is separable. However, the key difference is that a $\mathtt{MM}$ process $\mathsf{C}_{N:1}^{\textup{MM}}$ needs to be expressible as a convex combination of quantum channels (rather than states), which are operators that must respect an additional constraint. In particular, while quantum states satisfy $\tr{\rho}=1$, quantum channels must satisfy $\ptr{\out}{\mathsf{L}_{\out \inp}}=\mathbbm{1}_{\inp}$. This prevents one from directly employing the methods of entanglement theory to characterise $\mathtt{MM}$. Nonetheless, in Ref.~\cite{Nery_2021}, the authors present a method to construct inner and outer polytope approximations for $\mathtt{MM}$; similar methods have also been considered in Ref.~\cite{Ohst_2022}.

\vspace{-0.75em}\subsubsection*{\texorpdfstring{$\mathtt{CM}:$ }{} Open Problem}

Despite having an operational construction in terms of the underlying system-environment dynamics, checking whether a quantum process $\mathsf{C}_{N:1}\in\mathtt{QM}$ is a classical memory one (i.e., belongs to $\mathtt{CM}$) appears to be a very difficult problem. Loosely speaking, the difficulty of said task is akin to that of determining whether a channel can be implemented via LOCC: here, one needs to both track the entire history and deduce whether the dynamics could potentially be implemented by some sequence of conditional instruments. In this work, we have exploited the relationship between classical memory process and both non-signalling and separable ones in order to obtain non-trivial statements on $\mathtt{CM}$; nonetheless, systematic and efficient converging methods that provide both inner and outer approximations remain elusive. 

To this end, in App.~\ref{app::convexity}, we show that $\mathtt{CM}$ is a convex set. Hence, there exists a family of polytopes allowing arbitrarily good inner and outer approximations. Moreover, one can always ``witness'' processes lying outside $\mathtt{CM}$ via a hyperplane through the space of quantum processes. Constructing such witnesses or tractable characterisations for the set of classical memory quantum processes remains an important open question. In principle, one might employ a polytope / SDP hierarchy approach, in similar vein to the MM hierarchy presented in Ref.~\cite{Nery_2021}. However, since the set of classical memory quantum processes involves extremely non-trivial constraints on the decomposition elements, for $N>2$ time process, such a characterisation would have a high complexity cost and become intractable even for qubit scenarios.

\section{Concluding Discussion \& Outlook}

We have systematically divided quantum processes into classes that display remarkably distinct behaviours. First, in Thm.~\ref{thm::hierarchy}, we demonstrated a strict hierarchy based upon differing memory effects---none $\mathtt{M}$, three variants of classical memory $\mathtt{MM}, \mathtt{CM}$, and $\mathtt{SEP}$, and quantum memory $\mathtt{QM}$. In order to prove this hierarchy, we made use of the auxiliary set of non-signalling processes $\mathtt{NS}$, which are of significant interest in their own right. Upon deeper analysis, we showed that, although for two-time processes $\mathtt{CM} \subsetneq \mathtt{SEP} \subsetneq \mathtt{NS}$, there does not exist a hierarchical relationship between the sets in general (see Fig.~\ref{fig::hierarchy}). Our work therefore illuminates genuinely multi-time phenomena of open quantum processes.

Our results are of central importance in two ways. From a foundational standpoint, they distinguish quantum and classical memory, outlining the ultimate limitations of quantum information processing and providing a holistic characterisation of the hierarchy of possible memory structures in quantum theory. On the practical side, since noise in quantum devices---and thus the observed memory effects---will predominately be classical in the near future, our work provides a methodological framework upon which efficient and reliable quantum devices can be built. Accordingly, the concepts explored and results presented here should have immediate impact on various fields of quantum science, including quantum information theory, optimal control, open quantum systems, and quantum foundations, to name but a few. 

\vspace{-0.75em}\begin{acknowledgments}
The authors thank Seiseki Akibue, Jessica Bavaresco, Kavan Modi, and Satoshi Yoshida for insightful discussions. P. T. acknowledges support from the Japan Society for the Promotion of Science (JSPS) Postdoctoral Fellowship for Research in Japan. M. M. acknowledges support from MEXT Quantum Leap Flagship Program (MEXT QLEAP) JPMXS0118069605 (Q-LEAP of MEXT 2018-2027) and Japan Society for the Promotion of Science (JSPS) KAKENHI Grant No. JSPS KAKENHI 21H03394. This project has received funding from the European Union’s Horizon Europe research and innovation programme under the Marie Sk{\l}odowska-Curie grant agreement No. 101068332. 
\end{acknowledgments}

\def\bibsection{\section*{References}} 


\begin{thebibliography}{73}%
\makeatletter
\providecommand \@ifxundefined [1]{%
 \@ifx{#1\undefined}
}%
\providecommand \@ifnum [1]{%
 \ifnum #1\expandafter \@firstoftwo
 \else \expandafter \@secondoftwo
 \fi
}%
\providecommand \@ifx [1]{%
 \ifx #1\expandafter \@firstoftwo
 \else \expandafter \@secondoftwo
 \fi
}%
\providecommand \natexlab [1]{#1}%
\providecommand \enquote  [1]{#1}%
\providecommand \bibnamefont  [1]{#1}%
\providecommand \bibfnamefont [1]{#1}%
\providecommand \citenamefont [1]{#1}%
\providecommand \href@noop [0]{\@secondoftwo}%
\providecommand \href [0]{\begingroup \@sanitize@url \@href}%
\providecommand \@href[1]{\@@startlink{#1}\@@href}%
\providecommand \@@href[1]{\endgroup#1\@@endlink}%
\providecommand \@sanitize@url [0]{\catcode `\\12\catcode `\$12\catcode
  `\&12\catcode `\#12\catcode `\^12\catcode `\_12\catcode `\%12\relax}%
\providecommand \@@startlink[1]{}%
\providecommand \@@endlink[0]{}%
\providecommand \url  [0]{\begingroup\@sanitize@url \@url }%
\providecommand \@url [1]{\endgroup\@href {#1}{\urlprefix }}%
\providecommand \urlprefix  [0]{URL }%
\providecommand \Eprint [0]{\href }%
\providecommand \doibase [0]{https://doi.org/}%
\providecommand \selectlanguage [0]{\@gobble}%
\providecommand \bibinfo  [0]{\@secondoftwo}%
\providecommand \bibfield  [0]{\@secondoftwo}%
\providecommand \translation [1]{[#1]}%
\providecommand \BibitemOpen [0]{}%
\providecommand \bibitemStop [0]{}%
\providecommand \bibitemNoStop [0]{.\EOS\space}%
\providecommand \EOS [0]{\spacefactor3000\relax}%
\providecommand \BibitemShut  [1]{\csname bibitem#1\endcsname}%
\let\auto@bib@innerbib\@empty
\bibitem [{\citenamefont {van Kampen}(2011)}]{vanKampen_2011}%
  \BibitemOpen
  \bibfield  {author} {\bibinfo {author} {\bibfnamefont {N.}~\bibnamefont {van
  Kampen}},\ }\href@noop {} {\emph {\bibinfo {title} {{Stochastic Processes in
  Physics and Chemistry}}}}\ (\bibinfo  {publisher} {Elsevier, New York},\
  \bibinfo {year} {2011})\BibitemShut {NoStop}%
\bibitem [{\citenamefont {Sayrin}\ \emph {et~al.}(2011)\citenamefont {Sayrin},
  \citenamefont {Dotsenko}, \citenamefont {Zhou}, \citenamefont {Peaudecerf},
  \citenamefont {Rybarczyk}, \citenamefont {Gleyzes}, \citenamefont {Rouchon},
  \citenamefont {Mirrahimi}, \citenamefont {Amini}, \citenamefont {Brune},
  \citenamefont {Raimond},\ and\ \citenamefont {Haroche}}]{Sayrin_2011}%
  \BibitemOpen
  \bibfield  {author} {\bibinfo {author} {\bibfnamefont {C.}~\bibnamefont
  {Sayrin}}, \bibinfo {author} {\bibfnamefont {I.}~\bibnamefont {Dotsenko}},
  \bibinfo {author} {\bibfnamefont {X.}~\bibnamefont {Zhou}}, \bibinfo {author}
  {\bibfnamefont {B.}~\bibnamefont {Peaudecerf}}, \bibinfo {author}
  {\bibfnamefont {T.}~\bibnamefont {Rybarczyk}}, \bibinfo {author}
  {\bibfnamefont {S.}~\bibnamefont {Gleyzes}}, \bibinfo {author} {\bibfnamefont
  {P.}~\bibnamefont {Rouchon}}, \bibinfo {author} {\bibfnamefont
  {M.}~\bibnamefont {Mirrahimi}}, \bibinfo {author} {\bibfnamefont
  {H.}~\bibnamefont {Amini}}, \bibinfo {author} {\bibfnamefont
  {M.}~\bibnamefont {Brune}}, \bibinfo {author} {\bibfnamefont {J.-M.}\
  \bibnamefont {Raimond}}, \ and\ \bibinfo {author} {\bibfnamefont
  {S.}~\bibnamefont {Haroche}},\ }\emph {\enquote {\bibinfo {title} {{Real-time
  quantum feedback prepares and stabilizes photon number states}},}\ }\href
  {\doibase 10.1038/nature10376} {\bibfield  {journal} {\bibinfo  {journal}
  {Nature}\ }\textbf {\bibinfo {volume} {477}},\ \bibinfo {pages} {73}
  (\bibinfo {year} {2011})},\ \Eprint {http://arxiv.org/abs/arXiv:1107.4027}
  {arXiv:1107.4027}\BibitemShut {NoStop}%
\bibitem [{\citenamefont {Magrini}\ \emph {et~al.}(2021)\citenamefont
  {Magrini}, \citenamefont {Rosenzweig}, \citenamefont {Bach}, \citenamefont
  {Deutschmann-Olek}, \citenamefont {Hofer}, \citenamefont {Hong},
  \citenamefont {Kiesel}, \citenamefont {Kugi},\ and\ \citenamefont
  {Aspelmeyer}}]{Magrini_2021}%
  \BibitemOpen
  \bibfield  {author} {\bibinfo {author} {\bibfnamefont {L.}~\bibnamefont
  {Magrini}}, \bibinfo {author} {\bibfnamefont {P.}~\bibnamefont {Rosenzweig}},
  \bibinfo {author} {\bibfnamefont {C.}~\bibnamefont {Bach}}, \bibinfo {author}
  {\bibfnamefont {A.}~\bibnamefont {Deutschmann-Olek}}, \bibinfo {author}
  {\bibfnamefont {S.~G.}\ \bibnamefont {Hofer}}, \bibinfo {author}
  {\bibfnamefont {S.}~\bibnamefont {Hong}}, \bibinfo {author} {\bibfnamefont
  {N.}~\bibnamefont {Kiesel}}, \bibinfo {author} {\bibfnamefont
  {A.}~\bibnamefont {Kugi}}, \ and\ \bibinfo {author} {\bibfnamefont
  {M.}~\bibnamefont {Aspelmeyer}},\ }\emph {\enquote {\bibinfo {title}
  {{Real-time optimal quantum control of mechanical motion at room
  temperature}},}\ }\href {https://doi.org/10.1038/s41586-021-03602-3}
  {\bibfield  {journal} {\bibinfo  {journal} {Nature}\ }\textbf {\bibinfo
  {volume} {595}},\ \bibinfo {pages} {373} (\bibinfo {year} {2021})},\ \Eprint
  {http://arxiv.org/abs/arXiv:2012.15188} {arXiv:2012.15188}\BibitemShut
  {NoStop}%
\bibitem [{\citenamefont {Grimsmo}(2015)}]{Grimsmo_2015}%
  \BibitemOpen
  \bibfield  {author} {\bibinfo {author} {\bibfnamefont {A.~L.}\ \bibnamefont
  {Grimsmo}},\ }\emph {\enquote {\bibinfo {title} {{Time-Delayed Quantum
  Feedback Control}},}\ }\href {\doibase 10.1103/PhysRevLett.115.060402}
  {\bibfield  {journal} {\bibinfo  {journal} {Phys. Rev. Lett.}\ }\textbf
  {\bibinfo {volume} {115}},\ \bibinfo {pages} {060402} (\bibinfo {year}
  {2015})},\ \Eprint {http://arxiv.org/abs/arXiv:1502.06959}
  {arXiv:1502.06959}\BibitemShut {NoStop}%
\bibitem [{\citenamefont {Luchnikov}\ \emph {et~al.}(2019)\citenamefont
  {Luchnikov}, \citenamefont {Vintskevich}, \citenamefont {Ouerdane},\ and\
  \citenamefont {Filippov}}]{Luchnikov_2019}%
  \BibitemOpen
  \bibfield  {author} {\bibinfo {author} {\bibfnamefont {I.~A.}\ \bibnamefont
  {Luchnikov}}, \bibinfo {author} {\bibfnamefont {S.~V.}\ \bibnamefont
  {Vintskevich}}, \bibinfo {author} {\bibfnamefont {H.}~\bibnamefont
  {Ouerdane}}, \ and\ \bibinfo {author} {\bibfnamefont {S.~N.}\ \bibnamefont
  {Filippov}},\ }\emph {\enquote {\bibinfo {title} {{Simulation Complexity of
  Open Quantum Dynamics: Connection with Tensor Networks}},}\ }\href
  {https://link.aps.org/doi/10.1103/PhysRevLett.122.160401} {\bibfield
  {journal} {\bibinfo  {journal} {Phys. Rev. Lett.}\ }\textbf {\bibinfo
  {volume} {122}},\ \bibinfo {pages} {160401} (\bibinfo {year} {2019})},\
  \Eprint {http://arxiv.org/abs/arXiv:1812.00043}
  {arXiv:1812.00043}\BibitemShut {NoStop}%
\bibitem [{\citenamefont {J\o{}rgensen}\ and\ \citenamefont
  {Pollock}(2019)}]{Jorgensen_2019}%
  \BibitemOpen
  \bibfield  {author} {\bibinfo {author} {\bibfnamefont {M.~R.}\ \bibnamefont
  {J\o{}rgensen}}\ and\ \bibinfo {author} {\bibfnamefont {F.~A.}\ \bibnamefont
  {Pollock}},\ }\emph {\enquote {\bibinfo {title} {{Exploiting the Causal
  Tensor Network Structure of Quantum Processes to Efficiently Simulate
  Non-Markovian Path Integrals}},}\ }\href
  {https://link.aps.org/doi/10.1103/PhysRevLett.123.240602} {\bibfield
  {journal} {\bibinfo  {journal} {Phys. Rev. Lett.}\ }\textbf {\bibinfo
  {volume} {123}},\ \bibinfo {pages} {240602} (\bibinfo {year} {2019})},\
  \Eprint {http://arxiv.org/abs/arXiv:1902.00315}
  {arXiv:1902.00315}\BibitemShut {NoStop}%
\bibitem [{\citenamefont {Banaszek}\ \emph {et~al.}(2004)\citenamefont
  {Banaszek}, \citenamefont {Dragan}, \citenamefont {Wasilewski},\ and\
  \citenamefont {Radzewicz}}]{Banaszek_2004}%
  \BibitemOpen
  \bibfield  {author} {\bibinfo {author} {\bibfnamefont {K.}~\bibnamefont
  {Banaszek}}, \bibinfo {author} {\bibfnamefont {A.}~\bibnamefont {Dragan}},
  \bibinfo {author} {\bibfnamefont {W.}~\bibnamefont {Wasilewski}}, \ and\
  \bibinfo {author} {\bibfnamefont {C.}~\bibnamefont {Radzewicz}},\ }\emph
  {\enquote {\bibinfo {title} {{Experimental Demonstration of
  Entanglement-Enhanced Classical Communication over a Quantum Channel with
  Correlated Noise}},}\ }\href {\doibase 10.1103/PhysRevLett.92.257901}
  {\bibfield  {journal} {\bibinfo  {journal} {Phys. Rev. Lett.}\ }\textbf
  {\bibinfo {volume} {92}},\ \bibinfo {pages} {257901} (\bibinfo {year}
  {2004})},\ \Eprint {http://arxiv.org/abs/arXiv:quant-ph/0403024}
  {arXiv:quant-ph/0403024}\BibitemShut {NoStop}%
\bibitem [{\citenamefont {Bavaresco}\ \emph {et~al.}(2021)\citenamefont
  {Bavaresco}, \citenamefont {Murao},\ and\ \citenamefont
  {Quintino}}]{Bavaresco_2021}%
  \BibitemOpen
  \bibfield  {author} {\bibinfo {author} {\bibfnamefont {J.}~\bibnamefont
  {Bavaresco}}, \bibinfo {author} {\bibfnamefont {M.}~\bibnamefont {Murao}}, \
  and\ \bibinfo {author} {\bibfnamefont {M.~T.}\ \bibnamefont {Quintino}},\
  }\emph {\enquote {\bibinfo {title} {{Strict Hierarchy between Parallel,
  Sequential, and Indefinite-Causal-Order Strategies for Channel
  Discrimination}},}\ }\href {\doibase 10.1103/PhysRevLett.127.200504}
  {\bibfield  {journal} {\bibinfo  {journal} {Phys. Rev. Lett.}\ }\textbf
  {\bibinfo {volume} {127}},\ \bibinfo {pages} {200504} (\bibinfo {year}
  {2021})},\ \Eprint {http://arxiv.org/abs/arXiv:2011.08300}
  {arXiv:2011.08300}\BibitemShut {NoStop}%
\bibitem [{\citenamefont {Chiribella}\ \emph {et~al.}(2008)\citenamefont
  {Chiribella}, \citenamefont {D'Ariano},\ and\ \citenamefont
  {Perinotti}}]{Chiribella_2008_PRL}%
  \BibitemOpen
  \bibfield  {author} {\bibinfo {author} {\bibfnamefont {G.}~\bibnamefont
  {Chiribella}}, \bibinfo {author} {\bibfnamefont {G.~M.}\ \bibnamefont
  {D'Ariano}}, \ and\ \bibinfo {author} {\bibfnamefont {P.}~\bibnamefont
  {Perinotti}},\ }\emph {\enquote {\bibinfo {title} {{Quantum Circuit
  Architecture}},}\ }\href {\doibase 10.1103/PhysRevLett.101.060401} {\bibfield
   {journal} {\bibinfo  {journal} {Phys. Rev. Lett.}\ }\textbf {\bibinfo
  {volume} {101}},\ \bibinfo {pages} {060401} (\bibinfo {year} {2008})},\
  \Eprint {http://arxiv.org/abs/arXiv:0712.1325} {arXiv:0712.1325}\BibitemShut
  {NoStop}%
\bibitem [{\citenamefont {Chiribella}\ \emph {et~al.}(2009)\citenamefont
  {Chiribella}, \citenamefont {D{'}Ariano},\ and\ \citenamefont
  {Perinotti}}]{Chiribella_2009}%
  \BibitemOpen
  \bibfield  {author} {\bibinfo {author} {\bibfnamefont {G.}~\bibnamefont
  {Chiribella}}, \bibinfo {author} {\bibfnamefont {G.~M.}\ \bibnamefont
  {D{'}Ariano}}, \ and\ \bibinfo {author} {\bibfnamefont {P.}~\bibnamefont
  {Perinotti}},\ }\emph {\enquote {\bibinfo {title} {{Theoretical framework for
  quantum networks}},}\ }\href {\doibase 10.1103/PhysRevA.80.022339} {\bibfield
   {journal} {\bibinfo  {journal} {Phys. Rev. A}\ }\textbf {\bibinfo {volume}
  {80}},\ \bibinfo {pages} {022339} (\bibinfo {year} {2009})},\ \Eprint
  {http://arxiv.org/abs/arXiv:0904.4483} {arXiv:0904.4483}\BibitemShut
  {NoStop}%
\bibitem [{\citenamefont {Mavadia}\ \emph {et~al.}(2018)\citenamefont
  {Mavadia}, \citenamefont {Edmunds}, \citenamefont {Hempel}, \citenamefont
  {Ball}, \citenamefont {Roy}, \citenamefont {Stace},\ and\ \citenamefont
  {Biercuk}}]{Mavadia_2018}%
  \BibitemOpen
  \bibfield  {author} {\bibinfo {author} {\bibfnamefont {S.}~\bibnamefont
  {Mavadia}}, \bibinfo {author} {\bibfnamefont {C.~L.}\ \bibnamefont
  {Edmunds}}, \bibinfo {author} {\bibfnamefont {C.}~\bibnamefont {Hempel}},
  \bibinfo {author} {\bibfnamefont {H.}~\bibnamefont {Ball}}, \bibinfo {author}
  {\bibfnamefont {F.}~\bibnamefont {Roy}}, \bibinfo {author} {\bibfnamefont
  {T.~M.}\ \bibnamefont {Stace}}, \ and\ \bibinfo {author} {\bibfnamefont
  {M.~J.}\ \bibnamefont {Biercuk}},\ }\emph {\enquote {\bibinfo {title}
  {{Experimental quantum verification in the presence of temporally correlated
  noise}},}\ }\href {https://doi.org/10.1038/s41534-017-0052-0} {\bibfield
  {journal} {\bibinfo  {journal} {npj Quantum Inf.}\ }\textbf {\bibinfo
  {volume} {4}},\ \bibinfo {pages} {7} (\bibinfo {year} {2018})},\ \Eprint
  {http://arxiv.org/abs/arXiv:1706.03787} {arXiv:1706.03787}\BibitemShut
  {NoStop}%
\bibitem [{\citenamefont {White}\ \emph {et~al.}(2020)\citenamefont {White},
  \citenamefont {Hill}, \citenamefont {Pollock}, \citenamefont {Hollenberg},\
  and\ \citenamefont {Modi}}]{White_2020}%
  \BibitemOpen
  \bibfield  {author} {\bibinfo {author} {\bibfnamefont {G.~A.~L.}\
  \bibnamefont {White}}, \bibinfo {author} {\bibfnamefont {C.~D.}\ \bibnamefont
  {Hill}}, \bibinfo {author} {\bibfnamefont {F.~A.}\ \bibnamefont {Pollock}},
  \bibinfo {author} {\bibfnamefont {L.~C.~L.}\ \bibnamefont {Hollenberg}}, \
  and\ \bibinfo {author} {\bibfnamefont {K.}~\bibnamefont {Modi}},\ }\emph
  {\enquote {\bibinfo {title} {{Demonstration of non-Markovian process
  characterisation and control on a quantum processor}},}\ }\href
  {https://doi.org/10.1038/s41467-020-20113-3} {\bibfield  {journal} {\bibinfo
  {journal} {Nat. Commun.}\ }\textbf {\bibinfo {volume} {11}},\ \bibinfo
  {pages} {6301} (\bibinfo {year} {2020})},\ \Eprint
  {http://arxiv.org/abs/arXiv:2004.14018} {arXiv:2004.14018}\BibitemShut
  {NoStop}%
\bibitem [{\citenamefont {Ac{\' i}n}\ \emph {et~al.}(2018)\citenamefont {Ac{\'
  i}n}, \citenamefont {Bloch}, \citenamefont {Buhrman}, \citenamefont
  {Calarco}, \citenamefont {Eichler}, \citenamefont {Eisert}, \citenamefont
  {Esteve}, \citenamefont {Gisin}, \citenamefont {Glaser}, \citenamefont
  {Jelezko}, \citenamefont {Kuhr}, \citenamefont {Lewenstein}, \citenamefont
  {Riedel}, \citenamefont {Schmidt}, \citenamefont {Thew}, \citenamefont
  {Wallraff}, \citenamefont {Walmsley},\ and\ \citenamefont
  {Wilhelm}}]{Acin_2018}%
  \BibitemOpen
  \bibfield  {author} {\bibinfo {author} {\bibfnamefont {A.}~\bibnamefont
  {Ac{\' i}n}}, \bibinfo {author} {\bibfnamefont {I.}~\bibnamefont {Bloch}},
  \bibinfo {author} {\bibfnamefont {H.}~\bibnamefont {Buhrman}}, \bibinfo
  {author} {\bibfnamefont {T.}~\bibnamefont {Calarco}}, \bibinfo {author}
  {\bibfnamefont {C.}~\bibnamefont {Eichler}}, \bibinfo {author} {\bibfnamefont
  {J.}~\bibnamefont {Eisert}}, \bibinfo {author} {\bibfnamefont
  {D.}~\bibnamefont {Esteve}}, \bibinfo {author} {\bibfnamefont
  {N.}~\bibnamefont {Gisin}}, \bibinfo {author} {\bibfnamefont {S.~J.}\
  \bibnamefont {Glaser}}, \bibinfo {author} {\bibfnamefont {F.}~\bibnamefont
  {Jelezko}}, \bibinfo {author} {\bibfnamefont {S.}~\bibnamefont {Kuhr}},
  \bibinfo {author} {\bibfnamefont {M.}~\bibnamefont {Lewenstein}}, \bibinfo
  {author} {\bibfnamefont {M.~F.}\ \bibnamefont {Riedel}}, \bibinfo {author}
  {\bibfnamefont {P.~O.}\ \bibnamefont {Schmidt}}, \bibinfo {author}
  {\bibfnamefont {R.}~\bibnamefont {Thew}}, \bibinfo {author} {\bibfnamefont
  {A.}~\bibnamefont {Wallraff}}, \bibinfo {author} {\bibfnamefont
  {I.}~\bibnamefont {Walmsley}}, \ and\ \bibinfo {author} {\bibfnamefont
  {F.~K.}\ \bibnamefont {Wilhelm}},\ }\emph {\enquote {\bibinfo {title} {{The
  quantum technologies roadmap: a European community view}},}\ }\href
  {https://dx.doi.org/10.1088/1367-2630/aad1ea} {\bibfield  {journal} {\bibinfo
   {journal} {New J. Phys.}\ }\textbf {\bibinfo {volume} {20}},\ \bibinfo
  {pages} {080201} (\bibinfo {year} {2018})},\ \Eprint
  {http://arxiv.org/abs/arXiv:1712.03773} {arXiv:1712.03773}\BibitemShut
  {NoStop}%
\bibitem [{\citenamefont {Li}\ \emph {et~al.}(2018)\citenamefont {Li},
  \citenamefont {Hall},\ and\ \citenamefont {Wiseman}}]{Li_2018}%
  \BibitemOpen
  \bibfield  {author} {\bibinfo {author} {\bibfnamefont {L.}~\bibnamefont
  {Li}}, \bibinfo {author} {\bibfnamefont {M.~J.~W.}\ \bibnamefont {Hall}}, \
  and\ \bibinfo {author} {\bibfnamefont {H.~M.}\ \bibnamefont {Wiseman}},\
  }\emph {\enquote {\bibinfo {title} {{Concepts of quantum non-Markovianity: A
  hierarchy}},}\ }\href {\doibase 10.1016/j.physrep.2018.07.001} {\bibfield
  {journal} {\bibinfo  {journal} {Phys. Rep.}\ }\textbf {\bibinfo {volume}
  {759}},\ \bibinfo {pages} {1} (\bibinfo {year} {2018})},\ \Eprint
  {http://arxiv.org/abs/1712.08879} {arXiv:1712.08879}\BibitemShut {NoStop}%
\bibitem [{\citenamefont {Taranto}(2020)}]{TarantoThesis}%
  \BibitemOpen
  \bibfield  {author} {\bibinfo {author} {\bibfnamefont {P.}~\bibnamefont
  {Taranto}},\ }\emph {\enquote {\bibinfo {title} {{Memory Effects in Quantum
  Processes}},}\ }\href {\doibase 10.1142/S0219749919410028} {\bibfield
  {journal} {\bibinfo  {journal} {Int. J. Quantum Inf.}\ }\textbf {\bibinfo
  {volume} {18}},\ \bibinfo {pages} {1941002} (\bibinfo {year} {2020})},\
  \Eprint {http://arxiv.org/abs/1909.05245} {arXiv:1909.05245}\BibitemShut
  {NoStop}%
\bibitem [{\citenamefont {Pollock}\ \emph
  {et~al.}(2018{\natexlab{a}})\citenamefont {Pollock}, \citenamefont
  {Rodr\'{\i}guez-Rosario}, \citenamefont {Frauenheim}, \citenamefont
  {Paternostro},\ and\ \citenamefont {Modi}}]{Pollock_2018_PRL}%
  \BibitemOpen
  \bibfield  {author} {\bibinfo {author} {\bibfnamefont {F.~A.}\ \bibnamefont
  {Pollock}}, \bibinfo {author} {\bibfnamefont {C.}~\bibnamefont
  {Rodr\'{\i}guez-Rosario}}, \bibinfo {author} {\bibfnamefont {T.}~\bibnamefont
  {Frauenheim}}, \bibinfo {author} {\bibfnamefont {M.}~\bibnamefont
  {Paternostro}}, \ and\ \bibinfo {author} {\bibfnamefont {K.}~\bibnamefont
  {Modi}},\ }\emph {\enquote {\bibinfo {title} {{Operational Markov Condition
  for Quantum Processes}},}\ }\href {\doibase 10.1103/PhysRevLett.120.040405}
  {\bibfield  {journal} {\bibinfo  {journal} {Phys. Rev. Lett.}\ }\textbf
  {\bibinfo {volume} {120}},\ \bibinfo {pages} {040405} (\bibinfo {year}
  {2018}{\natexlab{a}})},\ \Eprint {http://arxiv.org/abs/arXiv:1801.09811}
  {arXiv:1801.09811}\BibitemShut {NoStop}%
\bibitem [{\citenamefont {Pollock}\ \emph
  {et~al.}(2018{\natexlab{b}})\citenamefont {Pollock}, \citenamefont
  {Rodr\'{\i}guez-Rosario}, \citenamefont {Frauenheim}, \citenamefont
  {Paternostro},\ and\ \citenamefont {Modi}}]{Pollock_2018_PRA}%
  \BibitemOpen
  \bibfield  {author} {\bibinfo {author} {\bibfnamefont {F.~A.}\ \bibnamefont
  {Pollock}}, \bibinfo {author} {\bibfnamefont {C.}~\bibnamefont
  {Rodr\'{\i}guez-Rosario}}, \bibinfo {author} {\bibfnamefont {T.}~\bibnamefont
  {Frauenheim}}, \bibinfo {author} {\bibfnamefont {M.}~\bibnamefont
  {Paternostro}}, \ and\ \bibinfo {author} {\bibfnamefont {K.}~\bibnamefont
  {Modi}},\ }\emph {\enquote {\bibinfo {title} {{Non-Markovian quantum
  processes: Complete framework and efficient characterization}},}\ }\href
  {\doibase 10.1103/PhysRevA.97.012127} {\bibfield  {journal} {\bibinfo
  {journal} {Phys. Rev. A}\ }\textbf {\bibinfo {volume} {97}},\ \bibinfo
  {pages} {012127} (\bibinfo {year} {2018}{\natexlab{b}})},\ \Eprint
  {http://arxiv.org/abs/arXiv:1512.00589} {arXiv:1512.00589}\BibitemShut
  {NoStop}%
\bibitem [{\citenamefont {Milz}\ \emph
  {et~al.}(2020{\natexlab{a}})\citenamefont {Milz}, \citenamefont {Sakuldee},
  \citenamefont {Pollock},\ and\ \citenamefont {Modi}}]{Milz_2020_Quantum}%
  \BibitemOpen
  \bibfield  {author} {\bibinfo {author} {\bibfnamefont {S.}~\bibnamefont
  {Milz}}, \bibinfo {author} {\bibfnamefont {F.}~\bibnamefont {Sakuldee}},
  \bibinfo {author} {\bibfnamefont {F.~A.}\ \bibnamefont {Pollock}}, \ and\
  \bibinfo {author} {\bibfnamefont {K.}~\bibnamefont {Modi}},\ }\emph {\enquote
  {\bibinfo {title} {{Kolmogorov extension theorem for (quantum) causal
  modelling and general probabilistic theories}},}\ }\href
  {https://doi.org/10.22331/q-2020-04-20-255} {\bibfield  {journal} {\bibinfo
  {journal} {{Quantum}}\ }\textbf {\bibinfo {volume} {4}},\ \bibinfo {pages}
  {255} (\bibinfo {year} {2020}{\natexlab{a}})},\ \Eprint
  {http://arxiv.org/abs/arXiv:1712.02589} {arXiv:1712.02589}\BibitemShut
  {NoStop}%
\bibitem [{\citenamefont {Strasberg}\ and\ \citenamefont
  {D\'{\i}az}(2019)}]{Strasberg_2019}%
  \BibitemOpen
  \bibfield  {author} {\bibinfo {author} {\bibfnamefont {P.}~\bibnamefont
  {Strasberg}}\ and\ \bibinfo {author} {\bibfnamefont {M.~G.}\ \bibnamefont
  {D\'{\i}az}},\ }\emph {\enquote {\bibinfo {title} {{Classical quantum
  stochastic processes}},}\ }\href {\doibase 10.1103/PhysRevA.100.022120}
  {\bibfield  {journal} {\bibinfo  {journal} {Phys. Rev. A}\ }\textbf {\bibinfo
  {volume} {100}},\ \bibinfo {pages} {022120} (\bibinfo {year} {2019})},\
  \Eprint {http://arxiv.org/abs/arXiv:1905.03018}
  {arXiv:1905.03018}\BibitemShut {NoStop}%
\bibitem [{\citenamefont {Milz}\ \emph
  {et~al.}(2020{\natexlab{b}})\citenamefont {Milz}, \citenamefont {Egloff},
  \citenamefont {Taranto}, \citenamefont {Theurer}, \citenamefont {Plenio},
  \citenamefont {Smirne},\ and\ \citenamefont {Huelga}}]{Milz_2020_PRX}%
  \BibitemOpen
  \bibfield  {author} {\bibinfo {author} {\bibfnamefont {S.}~\bibnamefont
  {Milz}}, \bibinfo {author} {\bibfnamefont {D.}~\bibnamefont {Egloff}},
  \bibinfo {author} {\bibfnamefont {P.}~\bibnamefont {Taranto}}, \bibinfo
  {author} {\bibfnamefont {T.}~\bibnamefont {Theurer}}, \bibinfo {author}
  {\bibfnamefont {M.~B.}\ \bibnamefont {Plenio}}, \bibinfo {author}
  {\bibfnamefont {A.}~\bibnamefont {Smirne}}, \ and\ \bibinfo {author}
  {\bibfnamefont {S.~F.}\ \bibnamefont {Huelga}},\ }\emph {\enquote {\bibinfo
  {title} {{When Is a Non-Markovian Quantum Process Classical?}}}\ }\href
  {https://link.aps.org/doi/10.1103/PhysRevX.10.041049} {\bibfield  {journal}
  {\bibinfo  {journal} {Phys. Rev. X}\ }\textbf {\bibinfo {volume} {10}},\
  \bibinfo {pages} {041049} (\bibinfo {year} {2020}{\natexlab{b}})},\ \Eprint
  {http://arxiv.org/abs/arXiv:1907.05807} {arXiv:1907.05807}\BibitemShut
  {NoStop}%
\bibitem [{\citenamefont {Taranto}\ \emph
  {et~al.}(2019{\natexlab{a}})\citenamefont {Taranto}, \citenamefont {Pollock},
  \citenamefont {Milz}, \citenamefont {Tomamichel},\ and\ \citenamefont
  {Modi}}]{Taranto_2019_PRL}%
  \BibitemOpen
  \bibfield  {author} {\bibinfo {author} {\bibfnamefont {P.}~\bibnamefont
  {Taranto}}, \bibinfo {author} {\bibfnamefont {F.~A.}\ \bibnamefont
  {Pollock}}, \bibinfo {author} {\bibfnamefont {S.}~\bibnamefont {Milz}},
  \bibinfo {author} {\bibfnamefont {M.}~\bibnamefont {Tomamichel}}, \ and\
  \bibinfo {author} {\bibfnamefont {K.}~\bibnamefont {Modi}},\ }\emph {\enquote
  {\bibinfo {title} {{Quantum Markov Order}},}\ }\href
  {https://doi.org/10.1103/PhysRevLett.122.140401} {\bibfield  {journal}
  {\bibinfo  {journal} {Phys. Rev. Lett.}\ }\textbf {\bibinfo {volume} {122}},\
  \bibinfo {pages} {140401} (\bibinfo {year} {2019}{\natexlab{a}})},\ \Eprint
  {http://arxiv.org/abs/arXiv:1805.11341} {arXiv:1805.11341}\BibitemShut
  {NoStop}%
\bibitem [{\citenamefont {Taranto}\ \emph
  {et~al.}(2019{\natexlab{b}})\citenamefont {Taranto}, \citenamefont {Milz},
  \citenamefont {Pollock},\ and\ \citenamefont {Modi}}]{Taranto_2019_PRA}%
  \BibitemOpen
  \bibfield  {author} {\bibinfo {author} {\bibfnamefont {P.}~\bibnamefont
  {Taranto}}, \bibinfo {author} {\bibfnamefont {S.}~\bibnamefont {Milz}},
  \bibinfo {author} {\bibfnamefont {F.~A.}\ \bibnamefont {Pollock}}, \ and\
  \bibinfo {author} {\bibfnamefont {K.}~\bibnamefont {Modi}},\ }\emph {\enquote
  {\bibinfo {title} {{Structure of quantum stochastic processes with finite
  Markov order}},}\ }\href {https://doi.org/10.1103/PhysRevA.99.042108}
  {\bibfield  {journal} {\bibinfo  {journal} {Phys. Rev. A}\ }\textbf {\bibinfo
  {volume} {99}},\ \bibinfo {pages} {042108} (\bibinfo {year}
  {2019}{\natexlab{b}})},\ \Eprint {http://arxiv.org/abs/arXiv:1810.10809}
  {arXiv:1810.10809}\BibitemShut {NoStop}%
\bibitem [{\citenamefont {Taranto}\ \emph {et~al.}(2021)\citenamefont
  {Taranto}, \citenamefont {Pollock},\ and\ \citenamefont
  {Modi}}]{Taranto_2021}%
  \BibitemOpen
  \bibfield  {author} {\bibinfo {author} {\bibfnamefont {P.}~\bibnamefont
  {Taranto}}, \bibinfo {author} {\bibfnamefont {F.~A.}\ \bibnamefont
  {Pollock}}, \ and\ \bibinfo {author} {\bibfnamefont {K.}~\bibnamefont
  {Modi}},\ }\emph {\enquote {\bibinfo {title} {{Non-Markovian memory strength
  bounds quantum process recoverability}},}\ }\href
  {https://doi.org/10.1038/s41534-021-00481-4} {\bibfield  {journal} {\bibinfo
  {journal} {npj Quantum Inf.}\ }\textbf {\bibinfo {volume} {7}},\ \bibinfo
  {pages} {149} (\bibinfo {year} {2021})},\ \Eprint
  {http://arxiv.org/abs/arXiv:1907.12583} {arXiv:1907.12583}\BibitemShut
  {NoStop}%
\bibitem [{\citenamefont {White}\ \emph {et~al.}(2021)\citenamefont {White},
  \citenamefont {Pollock}, \citenamefont {Hollenberg}, \citenamefont {Hill},\
  and\ \citenamefont {Modi}}]{White_2021}%
  \BibitemOpen
  \bibfield  {author} {\bibinfo {author} {\bibfnamefont {G.~A.~L.}\
  \bibnamefont {White}}, \bibinfo {author} {\bibfnamefont {F.~A.}\ \bibnamefont
  {Pollock}}, \bibinfo {author} {\bibfnamefont {L.~C.~L.}\ \bibnamefont
  {Hollenberg}}, \bibinfo {author} {\bibfnamefont {C.~D.}\ \bibnamefont
  {Hill}}, \ and\ \bibinfo {author} {\bibfnamefont {K.}~\bibnamefont {Modi}},\
  }\emph {\enquote {\bibinfo {title} {{ From many-body to many-time
  physics}},}\ }\href {https://arxiv.org/abs/2107.13934} {\bibfield  {journal}
  {\bibinfo  {journal} {arXiv:2107.13934}\ } (\bibinfo {year}
  {2021})}\BibitemShut {NoStop}%
\bibitem [{\citenamefont {White}\ \emph {et~al.}(2022)\citenamefont {White},
  \citenamefont {Pollock}, \citenamefont {Hollenberg}, \citenamefont {Modi},\
  and\ \citenamefont {Hill}}]{White_2022}%
  \BibitemOpen
  \bibfield  {author} {\bibinfo {author} {\bibfnamefont {G.~A.~L.}\
  \bibnamefont {White}}, \bibinfo {author} {\bibfnamefont {F.~A.}\ \bibnamefont
  {Pollock}}, \bibinfo {author} {\bibfnamefont {L.~C.~L.}\ \bibnamefont
  {Hollenberg}}, \bibinfo {author} {\bibfnamefont {K.}~\bibnamefont {Modi}}, \
  and\ \bibinfo {author} {\bibfnamefont {C.~D.}\ \bibnamefont {Hill}},\ }\emph
  {\enquote {\bibinfo {title} {{Non-Markovian Quantum Process Tomography}},}\
  }\href {https://link.aps.org/doi/10.1103/PRXQuantum.3.020344} {\bibfield
  {journal} {\bibinfo  {journal} {PRX Quantum}\ }\textbf {\bibinfo {volume}
  {3}},\ \bibinfo {pages} {020344} (\bibinfo {year} {2022})},\ \Eprint
  {http://arxiv.org/abs/arXiv:2106.11722} {arXiv:2106.11722}\BibitemShut
  {NoStop}%
\bibitem [{\citenamefont {Guo}\ \emph {et~al.}(2021)\citenamefont {Guo},
  \citenamefont {Taranto}, \citenamefont {Liu}, \citenamefont {Hu},
  \citenamefont {Huang}, \citenamefont {Li},\ and\ \citenamefont
  {Guo}}]{Guo_2021}%
  \BibitemOpen
  \bibfield  {author} {\bibinfo {author} {\bibfnamefont {Y.}~\bibnamefont
  {Guo}}, \bibinfo {author} {\bibfnamefont {P.}~\bibnamefont {Taranto}},
  \bibinfo {author} {\bibfnamefont {B.-H.}\ \bibnamefont {Liu}}, \bibinfo
  {author} {\bibfnamefont {X.-M.}\ \bibnamefont {Hu}}, \bibinfo {author}
  {\bibfnamefont {Y.-F.}\ \bibnamefont {Huang}}, \bibinfo {author}
  {\bibfnamefont {C.-F.}\ \bibnamefont {Li}}, \ and\ \bibinfo {author}
  {\bibfnamefont {G.-C.}\ \bibnamefont {Guo}},\ }\emph {\enquote {\bibinfo
  {title} {{Experimental Demonstration of Instrument-Specific Quantum Memory
  Effects and Non-Markovian Process Recovery for Common-Cause Processes}},}\
  }\href {https://doi.org/10.1103/PhysRevLett.126.230401} {\bibfield  {journal}
  {\bibinfo  {journal} {Phys. Rev. Lett.}\ }\textbf {\bibinfo {volume} {126}},\
  \bibinfo {pages} {230401} (\bibinfo {year} {2021})},\ \Eprint
  {http://arxiv.org/abs/arXiv:2003.14045} {arXiv:2003.14045}\BibitemShut
  {NoStop}%
\bibitem [{\citenamefont {{Kretschmann}}\ and\ \citenamefont
  {{Werner}}(2005)}]{Kretschmann_2005}%
  \BibitemOpen
  \bibfield  {author} {\bibinfo {author} {\bibfnamefont {D.}~\bibnamefont
  {{Kretschmann}}}\ and\ \bibinfo {author} {\bibfnamefont {R.~F.}\ \bibnamefont
  {{Werner}}},\ }\emph {\enquote {\bibinfo {title} {{Quantum channels with
  memory}},}\ }\href {\doibase 10.1103/PhysRevA.72.062323} {\bibfield
  {journal} {\bibinfo  {journal} {Phys. Rev. A}\ }\textbf {\bibinfo {volume}
  {72}},\ \bibinfo {pages} {062323} (\bibinfo {year} {2005})},\ \Eprint
  {http://arxiv.org/abs/quant-ph/0502106} {arXiv:quant-ph/0502106}\BibitemShut
  {NoStop}%
\bibitem [{\citenamefont {Gutoski}\ and\ \citenamefont
  {Watrous}(2007)}]{Gutoski_2007}%
  \BibitemOpen
  \bibfield  {author} {\bibinfo {author} {\bibfnamefont {G.}~\bibnamefont
  {Gutoski}}\ and\ \bibinfo {author} {\bibfnamefont {J.}~\bibnamefont
  {Watrous}},\ }\emph {\enquote {\bibinfo {title} {{Toward a General Theory of
  Quantum Games}},}\ }in\ \href {\doibase 10.1145/1250790.1250873} {\emph
  {\bibinfo {booktitle} {Proceedings of the Thirty-Ninth Annual ACM Symposium
  on Theory of Computing}}},\ \bibinfo {series and number} {STOC '07}\
  (\bibinfo  {publisher} {Association for Computing Machinery},\ \bibinfo
  {address} {New York, NY, USA},\ \bibinfo {year} {2007})\ p.\ \bibinfo {pages}
  {565},\ \Eprint {http://arxiv.org/abs/quant-ph/0611234}
  {arXiv:quant-ph/0611234}\BibitemShut {NoStop}%
\bibitem [{\citenamefont {Oreshkov}\ \emph {et~al.}(2012)\citenamefont
  {Oreshkov}, \citenamefont {Costa},\ and\ \citenamefont
  {Brukner}}]{Oreshkov_2012}%
  \BibitemOpen
  \bibfield  {author} {\bibinfo {author} {\bibfnamefont {O.}~\bibnamefont
  {Oreshkov}}, \bibinfo {author} {\bibfnamefont {F.}~\bibnamefont {Costa}}, \
  and\ \bibinfo {author} {\bibfnamefont {{\v C}.}~\bibnamefont {Brukner}},\
  }\emph {\enquote {\bibinfo {title} {{Quantum correlations with no causal
  order}},}\ }\href {\doibase 10.1038/ncomms2076} {\bibfield  {journal}
  {\bibinfo  {journal} {Nat. Commun.}\ }\textbf {\bibinfo {volume} {3}},\
  \bibinfo {pages} {1092} (\bibinfo {year} {2012})},\ \Eprint
  {http://arxiv.org/abs/1105.4464} {arXiv:1105.4464}\BibitemShut {NoStop}%
\bibitem [{\citenamefont {Hardy}(2012)}]{Hardy_2012}%
  \BibitemOpen
  \bibfield  {author} {\bibinfo {author} {\bibfnamefont {L.}~\bibnamefont
  {Hardy}},\ }\emph {\enquote {\bibinfo {title} {{The operator tensor
  formulation of quantum theory}},}\ }\href {\doibase 10.1098/rsta.2011.0326}
  {\bibfield  {journal} {\bibinfo  {journal} {Philos. Trans. Royal Soc. A}\
  }\textbf {\bibinfo {volume} {370}},\ \bibinfo {pages} {3385} (\bibinfo {year}
  {2012})},\ \Eprint {http://arxiv.org/abs/1201.4390}
  {arXiv:1201.4390}\BibitemShut {NoStop}%
\bibitem [{\citenamefont {Hardy}(2016)}]{Hardy_2016}%
  \BibitemOpen
  \bibfield  {author} {\bibinfo {author} {\bibfnamefont {L.}~\bibnamefont
  {Hardy}},\ }\emph {\enquote {\bibinfo {title} {{ Operational General
  Relativity: Possibilistic, Probabilistic, and Quantum}},}\ }\href
  {https://arxiv.org/abs/1608.06940} {\bibfield  {journal} {\bibinfo  {journal}
  {arXiv:1608.06940}\ } (\bibinfo {year} {2016})}\BibitemShut {NoStop}%
\bibitem [{\citenamefont {Lindblad}(1979)}]{Lindblad_1979}%
  \BibitemOpen
  \bibfield  {author} {\bibinfo {author} {\bibfnamefont {G.}~\bibnamefont
  {Lindblad}},\ }\emph {\enquote {\bibinfo {title} {{Non-Markovian quantum
  stochastic processes and their entropy}},}\ }\href
  {https://doi.org/10.1007/BF01197883} {\bibfield  {journal} {\bibinfo
  {journal} {Commun. Math. Phys.}\ }\textbf {\bibinfo {volume} {65}},\ \bibinfo
  {pages} {281} (\bibinfo {year} {1979})}\BibitemShut {NoStop}%
\bibitem [{\citenamefont {Accardi}\ \emph {et~al.}(1982)\citenamefont
  {Accardi}, \citenamefont {Frigerio},\ and\ \citenamefont
  {Lewis}}]{Accardi_1982}%
  \BibitemOpen
  \bibfield  {author} {\bibinfo {author} {\bibfnamefont {L.}~\bibnamefont
  {Accardi}}, \bibinfo {author} {\bibfnamefont {A.}~\bibnamefont {Frigerio}}, \
  and\ \bibinfo {author} {\bibfnamefont {J.~T.}\ \bibnamefont {Lewis}},\ }\emph
  {\enquote {\bibinfo {title} {{Quantum Stochastic Processes}},}\ }\href
  {\doibase 10.2977/prims/1195184017} {\bibfield  {journal} {\bibinfo
  {journal} {Publ. Rest. Inst. Math. Sci.}\ }\textbf {\bibinfo {volume} {18}},\
  \bibinfo {pages} {97} (\bibinfo {year} {1982})}\BibitemShut {NoStop}%
\bibitem [{\citenamefont {Oeckl}(2003)}]{Oeckl_2003}%
  \BibitemOpen
  \bibfield  {author} {\bibinfo {author} {\bibfnamefont {R.}~\bibnamefont
  {Oeckl}},\ }\emph {\enquote {\bibinfo {title} {A ``general boundary''
  formulation for quantum mechanics and quantum gravity},}\ }\href {\doibase
  https://doi.org/10.1016/j.physletb.2003.08.043} {\bibfield  {journal}
  {\bibinfo  {journal} {Phys. Lett. B}\ }\textbf {\bibinfo {volume} {575}},\
  \bibinfo {pages} {318} (\bibinfo {year} {2003})},\ \Eprint
  {http://arxiv.org/abs/hep-th/0306025} {arXiv:hep-th/0306025}\BibitemShut
  {NoStop}%
\bibitem [{\citenamefont {Aharonov}\ \emph {et~al.}(2009)\citenamefont
  {Aharonov}, \citenamefont {Popescu}, \citenamefont {Tollaksen},\ and\
  \citenamefont {Vaidman}}]{Aharonov_2009}%
  \BibitemOpen
  \bibfield  {author} {\bibinfo {author} {\bibfnamefont {Y.}~\bibnamefont
  {Aharonov}}, \bibinfo {author} {\bibfnamefont {S.}~\bibnamefont {Popescu}},
  \bibinfo {author} {\bibfnamefont {J.}~\bibnamefont {Tollaksen}}, \ and\
  \bibinfo {author} {\bibfnamefont {L.}~\bibnamefont {Vaidman}},\ }\emph
  {\enquote {\bibinfo {title} {Multiple-time states and multiple-time
  measurements in quantum mechanics},}\ }\href
  {https://link.aps.org/doi/10.1103/PhysRevA.79.052110} {\bibfield  {journal}
  {\bibinfo  {journal} {Phys. Rev. A}\ }\textbf {\bibinfo {volume} {79}},\
  \bibinfo {pages} {052110} (\bibinfo {year} {2009})},\ \Eprint
  {http://arxiv.org/abs/0712.0320} {arXiv:0712.0320}\BibitemShut {NoStop}%
\bibitem [{\citenamefont {Cotler}\ and\ \citenamefont
  {Wilczek}(2016)}]{Cotler_2016}%
  \BibitemOpen
  \bibfield  {author} {\bibinfo {author} {\bibfnamefont {J.}~\bibnamefont
  {Cotler}}\ and\ \bibinfo {author} {\bibfnamefont {F.}~\bibnamefont
  {Wilczek}},\ }\emph {\enquote {\bibinfo {title} {Entangled histories},}\
  }\href {\doibase 10.1088/0031-8949/2016/T168/014004} {\bibfield  {journal}
  {\bibinfo  {journal} {Phys. Scr.}\ }\textbf {\bibinfo {volume} {2016}},\
  \bibinfo {pages} {014004} (\bibinfo {year} {2016})},\ \Eprint
  {http://arxiv.org/abs/1502.02480} {arXiv:1502.02480}\BibitemShut {NoStop}%
\bibitem [{\citenamefont {Portmann}\ \emph {et~al.}(2017)\citenamefont
  {Portmann}, \citenamefont {Matt}, \citenamefont {Mauerer}, \citenamefont
  {Renner},\ and\ \citenamefont {Tackmann}}]{Portmann_2017}%
  \BibitemOpen
  \bibfield  {author} {\bibinfo {author} {\bibfnamefont {C.}~\bibnamefont
  {Portmann}}, \bibinfo {author} {\bibfnamefont {C.}~\bibnamefont {Matt}},
  \bibinfo {author} {\bibfnamefont {U.}~\bibnamefont {Mauerer}}, \bibinfo
  {author} {\bibfnamefont {R.}~\bibnamefont {Renner}}, \ and\ \bibinfo {author}
  {\bibfnamefont {B.}~\bibnamefont {Tackmann}},\ }\emph {\enquote {\bibinfo
  {title} {{Causal Boxes: Quantum Information-Processing Systems Closed under
  Composition}},}\ }\href {\doibase 10.1109/TIT.2017.2676805} {\bibfield
  {journal} {\bibinfo  {journal} {IEEE Trans. Inf. Theory}\ }\textbf {\bibinfo
  {volume} {65}},\ \bibinfo {pages} {3277} (\bibinfo {year} {2017})},\ \Eprint
  {http://arxiv.org/abs/1512.02240} {arXiv:1512.02240}\BibitemShut {NoStop}%
\bibitem [{\citenamefont {Berk}\ \emph
  {et~al.}(2021{\natexlab{a}})\citenamefont {Berk}, \citenamefont {Garner},
  \citenamefont {Yadin}, \citenamefont {Modi},\ and\ \citenamefont
  {Pollock}}]{berk_resource_2021}%
  \BibitemOpen
  \bibfield  {author} {\bibinfo {author} {\bibfnamefont {G.~D.}\ \bibnamefont
  {Berk}}, \bibinfo {author} {\bibfnamefont {A.~J.~P.}\ \bibnamefont {Garner}},
  \bibinfo {author} {\bibfnamefont {B.}~\bibnamefont {Yadin}}, \bibinfo
  {author} {\bibfnamefont {K.}~\bibnamefont {Modi}}, \ and\ \bibinfo {author}
  {\bibfnamefont {F.~A.}\ \bibnamefont {Pollock}},\ }\emph {\enquote {\bibinfo
  {title} {Resource theories of multi-time processes: {A} window into quantum
  non-{Markovianity}},}\ }\href {\doibase 10.22331/q-2021-04-20-435} {\bibfield
   {journal} {\bibinfo  {journal} {Quantum}\ }\textbf {\bibinfo {volume} {5}},\
  \bibinfo {pages} {435} (\bibinfo {year} {2021}{\natexlab{a}})},\ \Eprint
  {http://arxiv.org/abs/1907.07003} {arXiv:1907.07003}\BibitemShut {NoStop}%
\bibitem [{\citenamefont {Berk}\ \emph
  {et~al.}(2021{\natexlab{b}})\citenamefont {Berk}, \citenamefont {Milz},
  \citenamefont {Pollock},\ and\ \citenamefont {Modi}}]{berk_extracting_2021}%
  \BibitemOpen
  \bibfield  {author} {\bibinfo {author} {\bibfnamefont {G.~D.}\ \bibnamefont
  {Berk}}, \bibinfo {author} {\bibfnamefont {S.}~\bibnamefont {Milz}}, \bibinfo
  {author} {\bibfnamefont {F.~A.}\ \bibnamefont {Pollock}}, \ and\ \bibinfo
  {author} {\bibfnamefont {K.}~\bibnamefont {Modi}},\ }\emph {\enquote
  {\bibinfo {title} {Extracting {Quantum} {Dynamical} {Resources}:
  {Consumption} of {Non}-{Markovianity} for {Noise} {Reduction}},}\ }\href
  {http://arxiv.org/abs/2110.02613} {\bibfield  {journal} {\bibinfo  {journal}
  {arXiv:2110.02613}\ } (\bibinfo {year} {2021}{\natexlab{b}})}\BibitemShut
  {NoStop}%
\bibitem [{\citenamefont {Giarmatzi}\ and\ \citenamefont
  {Costa}(2021)}]{Giarmatzi_2021}%
  \BibitemOpen
  \bibfield  {author} {\bibinfo {author} {\bibfnamefont {C.}~\bibnamefont
  {Giarmatzi}}\ and\ \bibinfo {author} {\bibfnamefont {F.}~\bibnamefont
  {Costa}},\ }\emph {\enquote {\bibinfo {title} {Witnessing quantum memory in
  non-{M}arkovian processes},}\ }\href {\doibase 10.22331/q-2021-04-26-440}
  {\bibfield  {journal} {\bibinfo  {journal} {{Quantum}}\ }\textbf {\bibinfo
  {volume} {5}},\ \bibinfo {pages} {440} (\bibinfo {year} {2021})},\ \Eprint
  {http://arxiv.org/abs/arXiv:1811.03722} {arXiv:1811.03722}\BibitemShut
  {NoStop}%
\bibitem [{\citenamefont {Nery}\ \emph {et~al.}(2021)\citenamefont {Nery},
  \citenamefont {Quintino}, \citenamefont {Gu{\'{e}}rin}, \citenamefont
  {Maciel},\ and\ \citenamefont {Vianna}}]{Nery_2021}%
  \BibitemOpen
  \bibfield  {author} {\bibinfo {author} {\bibfnamefont {M.}~\bibnamefont
  {Nery}}, \bibinfo {author} {\bibfnamefont {M.~T.}\ \bibnamefont {Quintino}},
  \bibinfo {author} {\bibfnamefont {P.~A.}\ \bibnamefont {Gu{\'{e}}rin}},
  \bibinfo {author} {\bibfnamefont {T.~O.}\ \bibnamefont {Maciel}}, \ and\
  \bibinfo {author} {\bibfnamefont {R.~O.}\ \bibnamefont {Vianna}},\ }\emph
  {\enquote {\bibinfo {title} {{Simple and maximally robust processes with no
  classical common-cause or direct-cause explanation}},}\ }\href
  {https://doi.org/10.22331/q-2021-09-09-538} {\bibfield  {journal} {\bibinfo
  {journal} {{Quantum}}\ }\textbf {\bibinfo {volume} {5}},\ \bibinfo {pages}
  {538} (\bibinfo {year} {2021})},\ \Eprint
  {http://arxiv.org/abs/arXiv:2101.11630} {arXiv:2101.11630}\BibitemShut
  {NoStop}%
\bibitem [{\citenamefont {van Kampen}(1998)}]{vanKampen_1998}%
  \BibitemOpen
  \bibfield  {author} {\bibinfo {author} {\bibfnamefont {N.}~\bibnamefont {van
  Kampen}},\ }\emph {\enquote {\bibinfo {title} {{Remarks on Non-Markov
  Processes}},}\ }\href {https://doi.org/10.1590/S0103-97331998000200003}
  {\bibfield  {journal} {\bibinfo  {journal} {Braz. J. Phys.}\ }\textbf
  {\bibinfo {volume} {28}},\ \bibinfo {pages} {90} (\bibinfo {year}
  {1998})}\BibitemShut {NoStop}%
\bibitem [{\citenamefont {Horodecki}\ \emph {et~al.}(2009)\citenamefont
  {Horodecki}, \citenamefont {Horodecki}, \citenamefont {Horodecki},\ and\
  \citenamefont {Horodecki}}]{Horodecki_2009}%
  \BibitemOpen
  \bibfield  {author} {\bibinfo {author} {\bibfnamefont {R.}~\bibnamefont
  {Horodecki}}, \bibinfo {author} {\bibfnamefont {P.}~\bibnamefont
  {Horodecki}}, \bibinfo {author} {\bibfnamefont {M.}~\bibnamefont
  {Horodecki}}, \ and\ \bibinfo {author} {\bibfnamefont {K.}~\bibnamefont
  {Horodecki}},\ }\emph {\enquote {\bibinfo {title} {Quantum entanglement},}\
  }\href {\doibase 10.1103/RevModPhys.81.865} {\bibfield  {journal} {\bibinfo
  {journal} {Rev. Mod. Phys.}\ }\textbf {\bibinfo {volume} {81}},\ \bibinfo
  {pages} {865} (\bibinfo {year} {2009})},\ \Eprint
  {http://arxiv.org/abs/quant-ph/0702225} {arXiv:quant-ph/0702225}\BibitemShut
  {NoStop}%
\bibitem [{\citenamefont {Jamio{\l}kowski}(1972)}]{Jamiolkowski_1972}%
  \BibitemOpen
  \bibfield  {author} {\bibinfo {author} {\bibfnamefont {A.}~\bibnamefont
  {Jamio{\l}kowski}},\ }\emph {\enquote {\bibinfo {title} {{Linear
  transformations which preserve trace and positive semidefiniteness of
  operators}},}\ }\href {\doibase 10.1016/0034-4877(72)90011-0} {\bibfield
  {journal} {\bibinfo  {journal} {Rep. Math. Phys.}\ }\textbf {\bibinfo
  {volume} {3}},\ \bibinfo {pages} {275} (\bibinfo {year} {1972})}\BibitemShut
  {NoStop}%
\bibitem [{\citenamefont {Choi}(1975)}]{Choi_1975}%
  \BibitemOpen
  \bibfield  {author} {\bibinfo {author} {\bibfnamefont {M.-D.}\ \bibnamefont
  {Choi}},\ }\emph {\enquote {\bibinfo {title} {{Completely positive linear
  maps on complex matrices}},}\ }\href {\doibase 10.1016/0024-3795(75)90075-0}
  {\bibfield  {journal} {\bibinfo  {journal} {Linear Algebra Its Appl.}\
  }\textbf {\bibinfo {volume} {10}},\ \bibinfo {pages} {285} (\bibinfo {year}
  {1975})}\BibitemShut {NoStop}%
\bibitem [{\citenamefont {Budini}(2018)}]{Budini_2018}%
  \BibitemOpen
  \bibfield  {author} {\bibinfo {author} {\bibfnamefont {A.~A.}\ \bibnamefont
  {Budini}},\ }\emph {\enquote {\bibinfo {title} {{Quantum Non-Markovian
  Processes Break Conditional Past-Future Independence}},}\ }\href
  {https://link.aps.org/doi/10.1103/PhysRevLett.121.240401} {\bibfield
  {journal} {\bibinfo  {journal} {Phys. Rev. Lett.}\ }\textbf {\bibinfo
  {volume} {121}},\ \bibinfo {pages} {240401} (\bibinfo {year} {2018})},\
  \Eprint {http://arxiv.org/abs/arXiv:1811.03448}
  {arXiv:1811.03448}\BibitemShut {NoStop}%
\bibitem [{\citenamefont {Budini}(2022)}]{Budini_2022}%
  \BibitemOpen
  \bibfield  {author} {\bibinfo {author} {\bibfnamefont {A.~A.}\ \bibnamefont
  {Budini}},\ }\emph {\enquote {\bibinfo {title} {{Quantum Non-Markovian
  Environment-to-System Backflows of Information: Nonoperational vs.
  Operational Approaches}},}\ }\href {\doibase 10.3390/e24050649} {\bibfield
  {journal} {\bibinfo  {journal} {Entropy}\ }\textbf {\bibinfo {volume} {24}},\
  \bibinfo {pages} {649} (\bibinfo {year} {2022})},\ \Eprint
  {http://arxiv.org/abs/arXiv:2205.03333} {arXiv:2205.03333}\BibitemShut
  {NoStop}%
\bibitem [{\citenamefont {Strasberg}\ \emph {et~al.}(2023)\citenamefont
  {Strasberg}, \citenamefont {Winter}, \citenamefont {Gemmer},\ and\
  \citenamefont {Wang}}]{Strasberg_2023}%
  \BibitemOpen
  \bibfield  {author} {\bibinfo {author} {\bibfnamefont {P.}~\bibnamefont
  {Strasberg}}, \bibinfo {author} {\bibfnamefont {A.}~\bibnamefont {Winter}},
  \bibinfo {author} {\bibfnamefont {J.}~\bibnamefont {Gemmer}}, \ and\ \bibinfo
  {author} {\bibfnamefont {J.}~\bibnamefont {Wang}},\ }\emph {\enquote
  {\bibinfo {title} {{Classicality, Markovianity, and local detailed balance
  from pure-state dynamics}},}\ }\href {\doibase 10.1103/PhysRevA.108.012225}
  {\bibfield  {journal} {\bibinfo  {journal} {Phys. Rev. A}\ }\textbf {\bibinfo
  {volume} {108}},\ \bibinfo {pages} {012225} (\bibinfo {year} {2023})},\
  \Eprint {http://arxiv.org/abs/2209.07977} {arXiv:2209.07977}\BibitemShut
  {NoStop}%
\bibitem [{\citenamefont {Taranto}\ \emph {et~al.}(2023)\citenamefont
  {Taranto}, \citenamefont {Elliott},\ and\ \citenamefont
  {Milz}}]{Taranto_2023}%
  \BibitemOpen
  \bibfield  {author} {\bibinfo {author} {\bibfnamefont {P.}~\bibnamefont
  {Taranto}}, \bibinfo {author} {\bibfnamefont {T.~J.}\ \bibnamefont
  {Elliott}}, \ and\ \bibinfo {author} {\bibfnamefont {S.}~\bibnamefont
  {Milz}},\ }\emph {\enquote {\bibinfo {title} {Hidden {Q}uantum {M}emory: {I}s
  {M}emory {T}here {W}hen {S}omebody {L}ooks?}}\ }\href
  {https://doi.org/10.22331/q-2023-04-27-991} {\bibfield  {journal} {\bibinfo
  {journal} {{Quantum}}\ }\textbf {\bibinfo {volume} {7}},\ \bibinfo {pages}
  {991} (\bibinfo {year} {2023})},\ \Eprint
  {http://arxiv.org/abs/arXiv:2204.08298} {arXiv:2204.08298}\BibitemShut
  {NoStop}%
\bibitem [{\citenamefont {Horodecki}\ \emph {et~al.}(2003)\citenamefont
  {Horodecki}, \citenamefont {Shor},\ and\ \citenamefont
  {Ruskai}}]{Horodecki_2003}%
  \BibitemOpen
  \bibfield  {author} {\bibinfo {author} {\bibfnamefont {M.}~\bibnamefont
  {Horodecki}}, \bibinfo {author} {\bibfnamefont {P.~W.}\ \bibnamefont {Shor}},
  \ and\ \bibinfo {author} {\bibfnamefont {M.~B.}\ \bibnamefont {Ruskai}},\
  }\emph {\enquote {\bibinfo {title} {{Entanglement Breaking Channels}},}\
  }\href {\doibase 10.1142/S0129055X03001709} {\bibfield  {journal} {\bibinfo
  {journal} {Rev. Math. Phys.}\ }\textbf {\bibinfo {volume} {15}},\ \bibinfo
  {pages} {629} (\bibinfo {year} {2003})},\ \Eprint
  {http://arxiv.org/abs/arXiv:quant-ph/0302031}
  {arXiv:quant-ph/0302031}\BibitemShut {NoStop}%
\bibitem [{\citenamefont {Brunner}\ \emph {et~al.}(2014)\citenamefont
  {Brunner}, \citenamefont {Cavalcanti}, \citenamefont {Pironio}, \citenamefont
  {Scarani},\ and\ \citenamefont {Wehner}}]{Brunner_2014}%
  \BibitemOpen
  \bibfield  {author} {\bibinfo {author} {\bibfnamefont {N.}~\bibnamefont
  {Brunner}}, \bibinfo {author} {\bibfnamefont {D.}~\bibnamefont {Cavalcanti}},
  \bibinfo {author} {\bibfnamefont {S.}~\bibnamefont {Pironio}}, \bibinfo
  {author} {\bibfnamefont {V.}~\bibnamefont {Scarani}}, \ and\ \bibinfo
  {author} {\bibfnamefont {S.}~\bibnamefont {Wehner}},\ }\emph {\enquote
  {\bibinfo {title} {Bell nonlocality},}\ }\href
  {https://link.aps.org/doi/10.1103/RevModPhys.86.419} {\bibfield  {journal}
  {\bibinfo  {journal} {Rev. Mod. Phys.}\ }\textbf {\bibinfo {volume} {86}},\
  \bibinfo {pages} {419} (\bibinfo {year} {2014})},\ \Eprint
  {http://arxiv.org/abs/arXiv:1303.2849} {arXiv:1303.2849}\BibitemShut
  {NoStop}%
\bibitem [{\citenamefont {Gorini}\ \emph {et~al.}(1976)\citenamefont {Gorini},
  \citenamefont {Kossakowski},\ and\ \citenamefont {Sudarshan}}]{Gorini_1976}%
  \BibitemOpen
  \bibfield  {author} {\bibinfo {author} {\bibfnamefont {V.}~\bibnamefont
  {Gorini}}, \bibinfo {author} {\bibfnamefont {A.}~\bibnamefont {Kossakowski}},
  \ and\ \bibinfo {author} {\bibfnamefont {E.~C.~G.}\ \bibnamefont
  {Sudarshan}},\ }\emph {\enquote {\bibinfo {title} {{Completely positive
  dynamical semigroups of N‐level systems}},}\ }\href
  {https://doi.org/10.1063/1.522979} {\bibfield  {journal} {\bibinfo  {journal}
  {J. Math. Phys.}\ }\textbf {\bibinfo {volume} {17}},\ \bibinfo {pages} {821}
  (\bibinfo {year} {1976})}\BibitemShut {NoStop}%
\bibitem [{\citenamefont {Lindblad}(1976)}]{Lindblad_1976}%
  \BibitemOpen
  \bibfield  {author} {\bibinfo {author} {\bibfnamefont {G.}~\bibnamefont
  {Lindblad}},\ }\emph {\enquote {\bibinfo {title} {On the generators of
  quantum dynamical semigroups},}\ }\href {\doibase 10.1007/BF01608499}
  {\bibfield  {journal} {\bibinfo  {journal} {Commun. Math. Phys.}\ }\textbf
  {\bibinfo {volume} {48}},\ \bibinfo {pages} {119} (\bibinfo {year}
  {1976})}\BibitemShut {NoStop}%
\bibitem [{\citenamefont {Costa}\ and\ \citenamefont
  {Shrapnel}(2016)}]{Costa_2016}%
  \BibitemOpen
  \bibfield  {author} {\bibinfo {author} {\bibfnamefont {F.}~\bibnamefont
  {Costa}}\ and\ \bibinfo {author} {\bibfnamefont {S.}~\bibnamefont
  {Shrapnel}},\ }\emph {\enquote {\bibinfo {title} {Quantum causal
  modelling},}\ }\href {\doibase 10.1088/1367-2630/18/6/063032} {\bibfield
  {journal} {\bibinfo  {journal} {New J. Phys.}\ }\textbf {\bibinfo {volume}
  {18}},\ \bibinfo {pages} {063032} (\bibinfo {year} {2016})},\ \Eprint
  {http://arxiv.org/abs/arXiv:1512.07106} {arXiv:1512.07106}\BibitemShut
  {NoStop}%
\bibitem [{\citenamefont {{Ried}}\ \emph {et~al.}(2015)\citenamefont {{Ried}},
  \citenamefont {{Agnew}}, \citenamefont {{Vermeyden}}, \citenamefont
  {{Janzing}}, \citenamefont {{Spekkens}},\ and\ \citenamefont
  {{Resch}}}]{Reid_15}%
  \BibitemOpen
  \bibfield  {author} {\bibinfo {author} {\bibfnamefont {K.}~\bibnamefont
  {{Ried}}}, \bibinfo {author} {\bibfnamefont {M.}~\bibnamefont {{Agnew}}},
  \bibinfo {author} {\bibfnamefont {L.}~\bibnamefont {{Vermeyden}}}, \bibinfo
  {author} {\bibfnamefont {D.}~\bibnamefont {{Janzing}}}, \bibinfo {author}
  {\bibfnamefont {R.~W.}\ \bibnamefont {{Spekkens}}}, \ and\ \bibinfo {author}
  {\bibfnamefont {K.~J.}\ \bibnamefont {{Resch}}},\ }\emph {\enquote {\bibinfo
  {title} {{A quantum advantage for inferring causal structure}},}\ }\href
  {\doibase 10.1038/nphys3266} {\bibfield  {journal} {\bibinfo  {journal} {Nat.
  Phys.}\ }\textbf {\bibinfo {volume} {11}},\ \bibinfo {pages} {414} (\bibinfo
  {year} {2015})},\ \Eprint {http://arxiv.org/abs/1406.5036}
  {arXiv:1406.5036}\BibitemShut {NoStop}%
\bibitem [{\citenamefont {Feix}\ and\ \citenamefont
  {Brukner}(2017)}]{Feix_2017}%
  \BibitemOpen
  \bibfield  {author} {\bibinfo {author} {\bibfnamefont {A.}~\bibnamefont
  {Feix}}\ and\ \bibinfo {author} {\bibfnamefont {{\v C}.}~\bibnamefont
  {Brukner}},\ }\emph {\enquote {\bibinfo {title} {Quantum superpositions of
  `common-cause' and `direct-cause' causal structures},}\ }\href
  {https://dx.doi.org/10.1088/1367-2630/aa9b1a} {\bibfield  {journal} {\bibinfo
   {journal} {New J. Phys.}\ }\textbf {\bibinfo {volume} {19}},\ \bibinfo
  {pages} {123028} (\bibinfo {year} {2017})},\ \Eprint
  {http://arxiv.org/abs/arXiv:1606.09241} {arXiv:1606.09241}\BibitemShut
  {NoStop}%
\bibitem [{\citenamefont {Rains}(1997)}]{Rains_1997}%
  \BibitemOpen
  \bibfield  {author} {\bibinfo {author} {\bibfnamefont {E.~M.}\ \bibnamefont
  {Rains}},\ }\emph {\enquote {\bibinfo {title} {Entanglement purification via
  separable superoperators},}\ }\href {https://arxiv.org/abs/quant-ph/9707002}
  {\bibfield  {journal} {\bibinfo  {journal} {arXiv:quant-ph/9707002}\ }
  (\bibinfo {year} {1997})}\BibitemShut {NoStop}%
\bibitem [{\citenamefont {Vedral}\ \emph {et~al.}(1997)\citenamefont {Vedral},
  \citenamefont {Plenio}, \citenamefont {Rippin},\ and\ \citenamefont
  {Knight}}]{Vedral_1997}%
  \BibitemOpen
  \bibfield  {author} {\bibinfo {author} {\bibfnamefont {V.}~\bibnamefont
  {Vedral}}, \bibinfo {author} {\bibfnamefont {M.~B.}\ \bibnamefont {Plenio}},
  \bibinfo {author} {\bibfnamefont {M.~A.}\ \bibnamefont {Rippin}}, \ and\
  \bibinfo {author} {\bibfnamefont {P.~L.}\ \bibnamefont {Knight}},\ }\emph
  {\enquote {\bibinfo {title} {{Quantifying Entanglement}},}\ }\href
  {https://link.aps.org/doi/10.1103/PhysRevLett.78.2275} {\bibfield  {journal}
  {\bibinfo  {journal} {Phys. Rev. Lett.}\ }\textbf {\bibinfo {volume} {78}},\
  \bibinfo {pages} {2275} (\bibinfo {year} {1997})},\ \Eprint
  {http://arxiv.org/abs/quant-ph/9702027} {arXiv:quant-ph/9702027}\BibitemShut
  {NoStop}%
\bibitem [{\citenamefont {Bennett}\ \emph {et~al.}(1999)\citenamefont
  {Bennett}, \citenamefont {DiVincenzo}, \citenamefont {Fuchs}, \citenamefont
  {Mor}, \citenamefont {Rains}, \citenamefont {Shor}, \citenamefont {Smolin},\
  and\ \citenamefont {Wootters}}]{Bennett_1999}%
  \BibitemOpen
  \bibfield  {author} {\bibinfo {author} {\bibfnamefont {C.~H.}\ \bibnamefont
  {Bennett}}, \bibinfo {author} {\bibfnamefont {D.~P.}\ \bibnamefont
  {DiVincenzo}}, \bibinfo {author} {\bibfnamefont {C.~A.}\ \bibnamefont
  {Fuchs}}, \bibinfo {author} {\bibfnamefont {T.}~\bibnamefont {Mor}}, \bibinfo
  {author} {\bibfnamefont {E.}~\bibnamefont {Rains}}, \bibinfo {author}
  {\bibfnamefont {P.~W.}\ \bibnamefont {Shor}}, \bibinfo {author}
  {\bibfnamefont {J.~A.}\ \bibnamefont {Smolin}}, \ and\ \bibinfo {author}
  {\bibfnamefont {W.~K.}\ \bibnamefont {Wootters}},\ }\emph {\enquote {\bibinfo
  {title} {Quantum nonlocality without entanglement},}\ }\href
  {https://link.aps.org/doi/10.1103/PhysRevA.59.1070} {\bibfield  {journal}
  {\bibinfo  {journal} {Phys. Rev. A}\ }\textbf {\bibinfo {volume} {59}},\
  \bibinfo {pages} {1070} (\bibinfo {year} {1999})},\ \Eprint
  {http://arxiv.org/abs/quant-ph/9804053} {arXiv:quant-ph/9804053}\BibitemShut
  {NoStop}%
\bibitem [{\citenamefont {Milz}\ \emph {et~al.}(2021)\citenamefont {Milz},
  \citenamefont {Spee}, \citenamefont {Xu}, \citenamefont {Pollock},
  \citenamefont {Modi},\ and\ \citenamefont {Gühne}}]{Milz_2021}%
  \BibitemOpen
  \bibfield  {author} {\bibinfo {author} {\bibfnamefont {S.}~\bibnamefont
  {Milz}}, \bibinfo {author} {\bibfnamefont {C.}~\bibnamefont {Spee}}, \bibinfo
  {author} {\bibfnamefont {Z.-P.}\ \bibnamefont {Xu}}, \bibinfo {author}
  {\bibfnamefont {F.~A.}\ \bibnamefont {Pollock}}, \bibinfo {author}
  {\bibfnamefont {K.}~\bibnamefont {Modi}}, \ and\ \bibinfo {author}
  {\bibfnamefont {O.}~\bibnamefont {Gühne}},\ }\emph {\enquote {\bibinfo
  {title} {{Genuine multipartite entanglement in time}},}\ }\href
  {https://www.doi.org/10.21468/SciPostPhys.10.6.141} {\bibfield  {journal}
  {\bibinfo  {journal} {SciPost Phys.}\ }\textbf {\bibinfo {volume} {10}},\
  \bibinfo {pages} {141} (\bibinfo {year} {2021})},\ \Eprint
  {http://arxiv.org/abs/arXiv:2011.09340} {arXiv:2011.09340}\BibitemShut
  {NoStop}%
\bibitem [{\citenamefont {Beckman}\ \emph {et~al.}(2001)\citenamefont
  {Beckman}, \citenamefont {Gottesman}, \citenamefont {Nielsen},\ and\
  \citenamefont {Preskill}}]{beckman_causal_2001}%
  \BibitemOpen
  \bibfield  {author} {\bibinfo {author} {\bibfnamefont {D.}~\bibnamefont
  {Beckman}}, \bibinfo {author} {\bibfnamefont {D.}~\bibnamefont {Gottesman}},
  \bibinfo {author} {\bibfnamefont {M.~A.}\ \bibnamefont {Nielsen}}, \ and\
  \bibinfo {author} {\bibfnamefont {J.}~\bibnamefont {Preskill}},\ }\emph
  {\enquote {\bibinfo {title} {Causal and localizable quantum operations},}\
  }\href {\doibase 10.1103/PhysRevA.64.052309} {\bibfield  {journal} {\bibinfo
  {journal} {Phys. Rev. A}\ }\textbf {\bibinfo {volume} {64}},\ \bibinfo
  {pages} {052309} (\bibinfo {year} {2001})},\ \Eprint
  {http://arxiv.org/abs/arXiv:quant-ph/0102043}
  {arXiv:quant-ph/0102043}\BibitemShut {NoStop}%
\bibitem [{\citenamefont {Piani}\ \emph {et~al.}(2006)\citenamefont {Piani},
  \citenamefont {Horodecki}, \citenamefont {Horodecki},\ and\ \citenamefont
  {Horodecki}}]{piani_properties_2006}%
  \BibitemOpen
  \bibfield  {author} {\bibinfo {author} {\bibfnamefont {M.}~\bibnamefont
  {Piani}}, \bibinfo {author} {\bibfnamefont {M.}~\bibnamefont {Horodecki}},
  \bibinfo {author} {\bibfnamefont {P.}~\bibnamefont {Horodecki}}, \ and\
  \bibinfo {author} {\bibfnamefont {R.}~\bibnamefont {Horodecki}},\ }\emph
  {\enquote {\bibinfo {title} {Properties of quantum nonsignaling boxes},}\
  }\href {\doibase 10.1103/PhysRevA.74.012305} {\bibfield  {journal} {\bibinfo
  {journal} {Phys. Rev. A}\ }\textbf {\bibinfo {volume} {74}},\ \bibinfo
  {pages} {012305} (\bibinfo {year} {2006})},\ \Eprint
  {http://arxiv.org/abs/arXiv:quant-ph/0505110}
  {arXiv:quant-ph/0505110}\BibitemShut {NoStop}%
\bibitem [{\citenamefont {{Peres}}(1996)}]{Peres_1996}%
  \BibitemOpen
  \bibfield  {author} {\bibinfo {author} {\bibfnamefont {A.}~\bibnamefont
  {{Peres}}},\ }\emph {\enquote {\bibinfo {title} {{Separability Criterion for
  Density Matrices}},}\ }\href {\doibase 10.1103/PhysRevLett.77.1413}
  {\bibfield  {journal} {\bibinfo  {journal} {Phys. Rev. Lett.}\ }\textbf
  {\bibinfo {volume} {77}},\ \bibinfo {pages} {1413} (\bibinfo {year}
  {1996})},\ \Eprint {http://arxiv.org/abs/quant-ph/9604005}
  {arXiv:quant-ph/9604005}\BibitemShut {NoStop}%
\bibitem [{\citenamefont {{Horodecki}}\ \emph {et~al.}(1996)\citenamefont
  {{Horodecki}}, \citenamefont {{Horodecki}},\ and\ \citenamefont
  {{Horodecki}}}]{Horodecki_1996}%
  \BibitemOpen
  \bibfield  {author} {\bibinfo {author} {\bibfnamefont {M.}~\bibnamefont
  {{Horodecki}}}, \bibinfo {author} {\bibfnamefont {P.}~\bibnamefont
  {{Horodecki}}}, \ and\ \bibinfo {author} {\bibfnamefont {R.}~\bibnamefont
  {{Horodecki}}},\ }\emph {\enquote {\bibinfo {title} {{Separability of mixed
  states: necessary and sufficient conditions}},}\ }\href
  {https://dx.doi.org/10.1016/S0375-9601(96)00706-2} {\bibfield  {journal}
  {\bibinfo  {journal} {Phys. Lett. A}\ }\textbf {\bibinfo {volume} {223}},\
  \bibinfo {pages} {1} (\bibinfo {year} {1996})},\ \Eprint
  {http://arxiv.org/abs/quant-ph/9605038} {arXiv:quant-ph/9605038}\BibitemShut
  {NoStop}%
\bibitem [{\citenamefont {{G{\"u}hne}}\ and\ \citenamefont
  {{T{\'o}th}}(2009)}]{Guhne_2009}%
  \BibitemOpen
  \bibfield  {author} {\bibinfo {author} {\bibfnamefont {O.}~\bibnamefont
  {{G{\"u}hne}}}\ and\ \bibinfo {author} {\bibfnamefont {G.}~\bibnamefont
  {{T{\'o}th}}},\ }\emph {\enquote {\bibinfo {title} {{Entanglement
  detection}},}\ }\href {\doibase 10.1016/j.physrep.2009.02.004} {\bibfield
  {journal} {\bibinfo  {journal} {Phys. Rep.}\ }\textbf {\bibinfo {volume}
  {474}},\ \bibinfo {pages} {1} (\bibinfo {year} {2009})},\ \Eprint
  {http://arxiv.org/abs/0811.2803} {arXiv:0811.2803}\BibitemShut {NoStop}%
\bibitem [{\citenamefont {{Doherty}}\ \emph {et~al.}(2004)\citenamefont
  {{Doherty}}, \citenamefont {{Parrilo}},\ and\ \citenamefont
  {{Spedalieri}}}]{Doherty_2004}%
  \BibitemOpen
  \bibfield  {author} {\bibinfo {author} {\bibfnamefont {A.~C.}\ \bibnamefont
  {{Doherty}}}, \bibinfo {author} {\bibfnamefont {P.~A.}\ \bibnamefont
  {{Parrilo}}}, \ and\ \bibinfo {author} {\bibfnamefont {F.~M.}\ \bibnamefont
  {{Spedalieri}}},\ }\emph {\enquote {\bibinfo {title} {{Complete family of
  separability criteria}},}\ }\href {\doibase 10.1103/PhysRevA.69.022308}
  {\bibfield  {journal} {\bibinfo  {journal} {Phys. Rev. A}\ }\textbf {\bibinfo
  {volume} {69}},\ \bibinfo {eid} {022308} (\bibinfo {year} {2004})},\ \Eprint
  {http://arxiv.org/abs/quant-ph/0308032} {arXiv:quant-ph/0308032}\BibitemShut
  {NoStop}%
\bibitem [{\citenamefont {Vedral}\ and\ \citenamefont
  {Plenio}(1998)}]{vedral_entanglement_1998}%
  \BibitemOpen
  \bibfield  {author} {\bibinfo {author} {\bibfnamefont {V.}~\bibnamefont
  {Vedral}}\ and\ \bibinfo {author} {\bibfnamefont {M.~B.}\ \bibnamefont
  {Plenio}},\ }\emph {\enquote {\bibinfo {title} {Entanglement measures and
  purification procedures},}\ }\href {\doibase 10.1103/PhysRevA.57.1619}
  {\bibfield  {journal} {\bibinfo  {journal} {Phys. Rev. A}\ }\textbf {\bibinfo
  {volume} {57}},\ \bibinfo {pages} {1619} (\bibinfo {year}
  {1998})}\BibitemShut {NoStop}%
\bibitem [{\citenamefont {Bisio}\ \emph {et~al.}(2011)\citenamefont {Bisio},
  \citenamefont {Chiribella}, \citenamefont {D'Ariano},\ and\ \citenamefont
  {Perinotti}}]{Bisio_2011}%
  \BibitemOpen
  \bibfield  {author} {\bibinfo {author} {\bibfnamefont {A.}~\bibnamefont
  {Bisio}}, \bibinfo {author} {\bibfnamefont {G.}~\bibnamefont {Chiribella}},
  \bibinfo {author} {\bibfnamefont {G.~M.}\ \bibnamefont {D'Ariano}}, \ and\
  \bibinfo {author} {\bibfnamefont {P.}~\bibnamefont {Perinotti}},\ }\emph
  {\enquote {\bibinfo {title} {{Quantum Networks: general theory and
  applications}},}\ }\href
  {http://www.physics.sk/aps/pub.php?y=2011&pub=aps-11-03} {\bibfield
  {journal} {\bibinfo  {journal} {Acta Phys. Slovaca}\ }\textbf {\bibinfo
  {volume} {61}},\ \bibinfo {pages} {273} (\bibinfo {year} {2011})},\ \Eprint
  {http://arxiv.org/abs/1601.04864} {arXiv:1601.04864}\BibitemShut {NoStop}%
\bibitem [{\citenamefont {Chiribella}\ and\ \citenamefont
  {Ebler}(2016)}]{Chiribella_2016}%
  \BibitemOpen
  \bibfield  {author} {\bibinfo {author} {\bibfnamefont {G.}~\bibnamefont
  {Chiribella}}\ and\ \bibinfo {author} {\bibfnamefont {D.}~\bibnamefont
  {Ebler}},\ }\emph {\enquote {\bibinfo {title} {Optimal quantum networks and
  one-shot entropies},}\ }\href {\doibase 10.1088/1367-2630/18/9/093053}
  {\bibfield  {journal} {\bibinfo  {journal} {New J. Phys.}\ }\textbf {\bibinfo
  {volume} {18}},\ \bibinfo {pages} {093053} (\bibinfo {year} {2016})},\
  \Eprint {http://arxiv.org/abs/1606.02394} {arXiv:1606.02394}\BibitemShut
  {NoStop}%
\bibitem [{\citenamefont {Skrzypczyk}\ and\ \citenamefont
  {Cavalcanti}(2023)}]{Skrzypczyk_2023}%
  \BibitemOpen
  \bibfield  {author} {\bibinfo {author} {\bibfnamefont {P.}~\bibnamefont
  {Skrzypczyk}}\ and\ \bibinfo {author} {\bibfnamefont {D.}~\bibnamefont
  {Cavalcanti}},\ }\href {\doibase 10.1088/978-0-7503-3343-6} {\emph {\bibinfo
  {title} {{Semidefinite Programming in Quantum Information Science}}}}\
  (\bibinfo  {publisher} {IOP Publishing},\ \bibinfo {address} {Bristol, UK},\
  \bibinfo {year} {2023})\ \Eprint {http://arxiv.org/abs/2306.11637}
  {arXiv:2306.11637}\BibitemShut {NoStop}%
\bibitem [{\citenamefont {Gurvits}(2003)}]{Gurvits_2003}%
  \BibitemOpen
  \bibfield  {author} {\bibinfo {author} {\bibfnamefont {L.}~\bibnamefont
  {Gurvits}},\ }\emph {\enquote {\bibinfo {title} {{Classical deterministic
  complexity of Edmonds' problem and quantum entanglement}},}\ }in\ \href
  {\doibase 10.1145/780542.780545} {\emph {\bibinfo {booktitle} {Proceedings of
  the thirty-fifth annual ACM symposium on Theory of computing}}}\ (\bibinfo
  {year} {2003})\ pp.\ \bibinfo {pages} {10--19},\ \Eprint
  {http://arxiv.org/abs/quant-ph/0303055} {arXiv:quant-ph/0303055}\BibitemShut
  {NoStop}%
\bibitem [{\citenamefont {{Ohst}}\ \emph {et~al.}(2024)\citenamefont {{Ohst}},
  \citenamefont {{Yu}}, \citenamefont {{G{\"u}hne}},\ and\ \citenamefont {{Chau
  Nguyen}}}]{Ohst_2022}%
  \BibitemOpen
  \bibfield  {author} {\bibinfo {author} {\bibfnamefont {T.-A.}\ \bibnamefont
  {{Ohst}}}, \bibinfo {author} {\bibfnamefont {X.-D.}\ \bibnamefont {{Yu}}},
  \bibinfo {author} {\bibfnamefont {O.}~\bibnamefont {{G{\"u}hne}}}, \ and\
  \bibinfo {author} {\bibfnamefont {H.}~\bibnamefont {{Chau Nguyen}}},\ }\emph
  {\enquote {\bibinfo {title} {{Certifying quantum separability with adaptive
  polytopes}},}\ }\href {https://www.doi.org/10.21468/SciPostPhys.16.3.063}
  {\bibfield  {journal} {\bibinfo  {journal} {SciPost Phys.}\ }\textbf
  {\bibinfo {volume} {16}},\ \bibinfo {pages} {063} (\bibinfo {year} {2024})},\
  \Eprint {http://arxiv.org/abs/arXiv:2011.09340}
  {arXiv:2011.09340}\BibitemShut {NoStop}%
\bibitem [{\citenamefont {Stinespring}(1955)}]{Stinespring_1955}%
  \BibitemOpen
  \bibfield  {author} {\bibinfo {author} {\bibfnamefont {W.~F.}\ \bibnamefont
  {Stinespring}},\ }\emph {\enquote {\bibinfo {title} {Positive functions on
  {{$C^{\ast}$}}-algebras},}\ }\href
  {https://doi.org/10.1090/S0002-9939-1955-0069403-4} {\bibfield  {journal}
  {\bibinfo  {journal} {Proc. Amer. Math. Soc.}\ }\textbf {\bibinfo {volume}
  {6}},\ \bibinfo {pages} {211} (\bibinfo {year} {1955})}\BibitemShut {NoStop}%
\end{thebibliography}
%

\vspace{3em}
\onecolumngrid \newpage
\appendix

\section{Classical Memory Quantum Processes}\label{app::classicalmemoryquantumprocesses}

Here we demonstrate that quantum processes with an EBC applied on the environment in between each pair of times are equivalent to the classical memory processes described in Ref.~\cite{Giarmatzi_2021} [see Eq.~\eqref{eq::classicalmemoryquantumprocess-2}]. 

We begin by showing that the system-environment construction formulated in Def.~\ref{def::classicalmemoryquantumprocess} leads to Choi operators of the form presented in Eq.~\eqref{eq::classicalmemoryquantumprocess-2}. To this end, we first decompose each EBC into a measure-and-prepare implementation
\begin{align}
    \mathsf{E}_{E_j} := \sum_{x_j} \sigma_{j^{\out}}^{(x_j)} \otimes \mathsf{M}_{j^{\inp}}^{(x_j)},
\end{align}
where each $\sigma_{j^{\out}}^{(x_j)}$ is a quantum state and $\{ \mathsf{M}_{j^{\inp}}^{(x_j)} \}$ represents a POVM~\cite{Horodecki_2003}. One can then define:
\begin{align}
    \rho_{1^{\inp}}^{(x_1)} &:=  \mathsf{M}_{E_{1^{\inp}}}^{(x_1)} \star \rho_{(ES)_{1^{\inp}}}, \label{eq::stateensemble} \\
    \mathsf{L}_{j+1^{\inp}:j^{\out}}^{(x_{j+1}|x_{j:1})} &:= \mathsf{M}_{E_{j+1^{\inp}}}^{(x_{j+1}|x_{j:1})} \star \mathsf{U}_{(ES)_{j+1^{\inp} j^{\out}}} \star \sigma_{E_{j^{\out}}}^{(x_j|x_{j-1:1})}, \label{eq::conditionalinstrument} \\
    \mathsf{L}_{N^{\inp}:N-1^{\out}}^{(|x_{N-1:1})} &:= \mathbbm{1}_{E_{N^{\inp}}} \star \mathsf{U}_{(ES)_{N^{\inp} N-1^{\out}}} \star \sigma_{E_{N-1^{\out}}}^{(x_{N-1}|x_{N-2:1})} \label{eq::finalchannel}.
\end{align}
It is easy to verify that by construction: Eq.~\eqref{eq::stateensemble} corresponds to a state ensemble, i.e., $\rho_{1^{\inp}}^{(x_1)} \geq 0$ and $\tr{\sum_{x_1}\rho_{1^{\inp}}^{(x_1)}} = 1$; Eq.~\eqref{eq::conditionalinstrument} corresponds to a conditional instrument, i.e., $\mathsf{L}_{j+1^{\inp}:j^{\out}}^{(x_{j+1}|x_{j:1})} \geq 0$ and $\ptr{j+1^{\inp}}{\sum_{x_{j+1}} \mathsf{L}_{j+1^{\inp}:j^{\out}}^{(x_{j+1}|x_{j:1})}} = \mathbbm{1}_{j^{\out}} \; \forall \; x_{j:1}$ (since each $\sigma_{E_{j^{\out}}}^{(x_j|x_{j-1:1})}$ is a normalised quantum state for all $x_{j:1}$ and $\{\mathsf{M}_{E_{j+1^{\inp}}}^{(x_{j+1}|x_{j:1})}\}$ forms a POVM over the potential outcomes $x_{j+1}$ for all $x_{j:1}$); and Eq.~\eqref{eq::finalchannel} corresponds to a CPTP channel for every previous observation sequence, i.e., $\mathsf{L}_{N^{\inp}:N-1^{\out}}^{(|x_{N-1:1})} \geq 0$ and $\ptr{N^{\inp}}{\mathsf{L}_{N^{\inp}:N-1^{\out}}^{(|x_{N-1:1})}} = \mathbbm{1}_{N-1^{\out}} \; \forall \; x_{N-1:1}$ (since each $\sigma_{E_{N-1^{\out}}}^{(x_{N-1}|x_{N-1:1})}$ is a normalised quantum state). The link product preserves positivity of operators~\cite{Chiribella_2009}; furthermore, since each object above is defined upon independent spaces, the link product reduces to the tensor product. Hence, it follows that the construction of $\mathsf{C}_{N:1}^{\textup{CM}}$ according to Eq.~\eqref{eq::classicalmemoryquantumprocess-2} is a positive semidefinite operator. Lastly, it is also straightforward to check that it satisfies the causality conditions outlined in Eq.~\eqref{eq::causality_conditions}, thereby constituting a valid quantum process.

In the converse direction, note that any state ensemble $\{ \rho_{1^{\inp}}^{(x_1)} \}$ can be constructed via a dilation to a state on a larger system-environment space $\rho_{(ES)_{1^{\inp}}}$ and a measurement $\mathsf{M}_{E_{1^{\inp}}}^{(x_1)}$ on the environment; any (conditional) instrument can be realised similarly in terms of an initially-prepared environment state (which can depend upon the entire previous sequence of outcomes) $\sigma_{E_{j^{\out}}}^{(x_j|x_{j-1:1})}$, a global unitary dynamics $\mathsf{U}_{(ES)_{j+1^{\inp} j^{\out}}}$ and a measurement (whose choice can depend upon the previous outcomes) $\mathsf{M}_{E_{j+1^{\inp}}}^{(x_{j+1}|x_{j:1})}$ on the environment; and lastly that any channel can be realised in terms of an initially-prepared environment state (which can depend upon the entire previous sequence of outcomes) $\sigma_{E_{{N-1}^{\out}}}^{(x_{N-1}|x_{N-2:1})}$, a global unitary dynamics $\mathsf{U}_{(ES)_{N^{\inp} N-1^{\out}}}$ and a final trace over the environment $\mathbbm{1}_{E_{N^{\inp}}}$. In other words, the r.h.s. of Eqs.~\eqref{eq::stateensemble}--\eqref{eq::finalchannel} always provides a potential realisation of the l.h.s., which is essentially the Stinespring representation of probabilistic quantum transformations~\cite{Stinespring_1955}. Concatenating these elementary portions of dynamics thus allows one to represent any quantum process in Eq.~\eqref{eq::classicalmemoryquantumprocess-2} to system-environment dynamics of the explicit form presented in Def.~\ref{def::classicalmemoryquantumprocess}.

\section{Two-Time Equivalences Between Memory Classes}\label{app::twotimeequivalence}

\begin{thm}\label{thm::hierarchy-2}
For $N=2$ times, we have the relations:
\begin{gather}
    \mathtt{M} \subsetneq \mathtt{MM} = \mathtt{CM} \subsetneq \mathtt{SEP} \subsetneq \mathtt{NS} = \mathtt{QM}.
\end{gather}
\end{thm}

First of all, we will reiterate the proofs that show that, for processes defined upon two times, classical memory and mixed memoryless quantum processes are equivalent. Note that $\mathtt{MM} = \mathtt{CM}$ for $N=2$ has been shown in Refs.~\cite{Giarmatzi_2021,Nery_2021}, but we include it here for the sake of completeness.

\begin{lem}\label{lem::MMeqcm-2}
For $N=2$ times, mixed memoryless quantum processes coincide with classical memory quantum processes, i.e., $\mathtt{MM} = \mathtt{CM}$. 
\end{lem}

\begin{proof}
We first show that any two-time classical memory quantum process can be written in the form
\begin{align}\label{eq::initialseparablestate}
    \mathsf{C}_{2:1}^{\textup{CM}} = \rho_{E S_{1^{\inp}}}^{\textup{SEP}} \star \mathsf{D}_{E S_{1^{\out}} S_{2^{\inp}}},
\end{align}
where $\rho_{E S_{1^{\inp}}}^{\textup{SEP}}$ is a separable state and $\mathsf{D}_{E S_{1^{\out}} S_{2^{\inp}}}$ is (the Choi operator of) a proper quantum channel $\mathcal{L}(\mathcal{H}_E \otimes \mathcal{H}_{S_{1^{\out}}}) \mapsto \mathcal{L}(\mathcal{H}_{S_{2^{\inp}}})$. Here, the link product $\star$ simply denotes that the environment subsystem of the initial state feeds forward as an input into the dynamics between times. We will refer to such processes as ``initial separable state'' processes and use them as an intermediary to prove that $\mathtt{MM} = \mathtt{CM}$.

To see that an classical memory quantum process leads to the form of Eq.~\eqref{eq::initialseparablestate}, note that by considering an EBC applied to the environment between times (which we split into the input and output spaces of the channel), $\mathsf{E}_{E_1} = \sum_x \mathsf{M}_{E_{1^{\inp}}}^{(x)} \otimes \sigma_{E_{1^{\out}}}^{(x)}$, one has
\begin{align}
    \mathsf{C}_{2:1}^{\textup{CM}} &= \sum_x \rho_{S_{1^{\inp}} E_{1^{\inp}}} \star \mathsf{M}_{E_{1^{\inp}}}^{(x)} \otimes \sigma_{E_{1^{\out}}}^{(x)} \star \mathsf{D}_{E_{1^{\out}} S_{1^{\out}} S_{2^{\inp}}} \notag \\
    &= \sum_x p_x \rho_{S_{1^{\inp}}}^{(x)} \otimes \sigma_{E_{1^{\out}}}^{(x)} \star \mathsf{D}_{E_{1^{\out}} S_{1^{\out}} S_{2^{\inp}}} = \rho_{E_{1^{\out}} S_{1^{\inp}}}^{\textup{SEP}} \star \mathsf{D}_{E_{1^{\out}} S_{1^{\out}} S_{2^{\inp}}},
\end{align}
where we defined $p_x \rho_{S_{1^{\inp}}}^{(x)} := \rho_{S_{1^{\inp}} E_{1^{\inp}}} \star \mathsf{M}_{E_{1^{\inp}}}^{(x)}$. This asserts the claim since each $\sigma_{E_{1^{\out}}}^{(x)}$ is a valid quantum state, and so the initial state of the process is indeed separable, and furthermore $\mathsf{D}_{E_{1^{\out}} S_{1^{\out}} S_{2^{\inp}}}$ is a valid quantum channel. 

In the converse direction, beginning with an initial separable state process, one has
\begin{align}
    \rho_{E S_{1^{\inp}}}^{\textup{SEP}} \star \mathsf{D}_{E S_{1^{\out}} S_{2^{\inp}}} &= \sum_{x} p_x \rho_{S_{1^{\inp}}}^{(x)} \otimes \rho_{E}^{(x)} \star \mathsf{D}_{E S_{1^{\out}} S_{2^{\inp}}}.
\end{align}
Since any reduced ensemble can arise from the measurement of part of some larger system, we write $\rho_{S_{1^{\inp}}}^{(x)} = \ptr{A}{\mathsf{M}_A^{(x)} \rho_{A S_{1^{\inp}}}}$ for some Hilbert space $\mathcal{H}_A$ and POVM $\{ \mathsf{M}^{(x)}\}$. Relabelling the systems $A \mapsto E_{1^{\inp}}$ and $E \mapsto E_{1^{\out}}$, we then have a classical memory process construction in terms of the EBC on the environment $\mathsf{E}_{E_1} = \sum_x \mathsf{M}_{E_{1^{\inp}}}^{(x)} \otimes \rho_{E_{1^{\out}}}^{(x)}$.

In addition to the equivalence between classical memory quantum processes and those with an initially separable states, we now show that these are equivalent to mixed memoryless processes. To this end, we need to show that
\begin{align}
    \rho_{E S_{1^{\inp}}}^{\textup{SEP}} \star \mathsf{D}_{E S_{1^{\out}} S_{2^{\inp}}} \Leftrightarrow \sum_{x} p_x \mathsf{L}_{2^{\inp}:1^{\out}}^{(x)} \otimes \rho_{1^{\inp}}^{(x)}.
\end{align}
In the forwards direction, one can simply expand the separable state as $\rho_{E S_{1^{\inp}}}^{\textup{SEP}} = \sum_{x} p_x \rho_{S_{1^{\inp}}}^{(x)} \otimes \rho_{E}^{(x)}$ and explicitly calculate the link product
to yield $\rho_E^{(x)} \star \mathsf{D}_{E S_{1^{\out}} S_{2^{\inp}}} = \mathsf{L}_{2^{\inp}:1^{\out}}^{(x)}$, which is CPTP for each $x$ by construction. 

To prove the converse direction, one can define
\begin{align}
    \rho_{E S_{1^{\inp}}}^{\textup{SEP}} := \sum_x p_x \rho^{(x)}_{1^{\inp}} \otimes \ketbra{x}{x}_{E} \notag \\
    \mathsf{D}_{E S_{1^{\out}} S_{2^{\inp}}} := \sum_x \ketbra{x}{x}_{E} \otimes \mathsf{L}_{2^{\inp}:1^{\out}}^{(x)}.
\end{align}
It is straightforward to show that $\rho_{E S_{1^{\inp}}}^{\textup{SEP}} \star  \mathsf{D}_{E S_{1^{\out}} S_{2^{\inp}}} = \sum_{x} p_x \mathsf{L}_{2^{\inp}:1^{\out}}^{(x)} \otimes \rho_{1^{\inp}}^{(x)}$ holds true, asserting the claim.
\end{proof}

The next part of Thm.~\ref{thm::hierarchy-2} makes use of:
\begin{lem}[\cite{Nery_2021}]\label{lem::sepnotcm-2}
For $N=2$ times, mixed memoryless quantum processes are a strict subset of separable quantum processes, i.e., $\mathtt{MM} \subsetneq \mathtt{SEP}$. 
\end{lem}

The above lemma is proven by the explicit example presented in Eq.~\eqref{eq::guerin} of the main text, which was first presented in Ref.~\cite{Nery_2021}. This process is separable (by construction) and so necessarily evades detection by the entanglement witnesses for quantum processes developed in Ref.~\cite{Giarmatzi_2021}. Nonetheless, it \emph{cannot} be realised by a classical memory quantum process, demonstrating that separability (plus overall causality) is not strong enough to enforce that the individual elements in the separable decomposition can be written as a sequence of conditional instruments, even for processes defined upon only two times~\cite{Nery_2021}.

\begin{lem}\label{lem::nseqqm-2}
For $N=2$ times, non-signalling quantum processes coincide with general quantum processes, i.e., $\mathtt{NS} = \mathtt{QM}$. 
\end{lem}

\begin{proof}
    Simply notice that for $N=2$ times, there is only one non-signalling constraint, which is automatically already ensured by the causality constraint of a general quantum process.
\end{proof}

Lastly, it is clear that for $N=2$ times, $\mathtt{SEP} \subsetneq \mathtt{NS} = \mathtt{QM}$ since processes in the latter set can be entangled. This completes the proof of Thm.~\ref{thm::hierarchy-2}.

\section{Classical Memory Quantum Process with no Mixed Memoryless Quantum Process Realisation}\label{app::cmnotMM}

Here we prove the strict relation $\mathtt{MM} \subsetneq \mathtt{CM}$, as claimed in Lem.~\ref{lem::DCsubsetofCM}. We first derive a necessary condition that must be satisfied by all mixed memoryless processes, namely that they can only exhibit signalling from any given output time to the next input, i.e., we first prove Lem.~\ref{lem::MMnosignalling}. To complete the proof of Lem.~\ref{lem::DCsubsetofCM}, we construct an explicit example of a classical memory process that violates the aforementioned condition (i.e., Lem.~\ref{lem::cmbutsignalling}), thereby asserting the claim. 

\begin{proof}[Proof: Lem.~\ref{lem::MMnosignalling}.]
Begin with the general form of a $\mathtt{MM}$ process and compute 
\begin{align}
    &\mathrm{tr}_{j+1^{\inp}}[\mathsf{C}_{N:1}^{\textup{MM}}]
    = \sum_x p_x \mathrm{tr}_{j+1^{\inp}}[\bigotimes_{i=1}^{N-1} L_{i+1^{\inp}:i^{\out}}^{(x)} \otimes \rho_{1^{\inp}}^{(x)}] \notag \\
    &= \sum_x p_x \left(\bigotimes_{i=j+1}^{N-1} L_{i+1^{\inp}:i^{\out}}^{(x)} \right) \otimes \mathbbm{1}_{j^{\out}} \left(\bigotimes_{i=1}^{j-1} L_{i+1^{\inp}:i^{\out}}^{(x)} \otimes \rho_{1^{\inp}}^{(x)}\right) \notag \\
    &= \mathbbm{1}_{j^{\out}} \otimes \mathsf{C}^{\textup{MM}}_{N:1 \setminus j+1^{\out}j^{\inp}}.
\end{align}
To derive the second equality, we made use of the fact that each $L_{j+1^{\inp}:j^{\out}}^{(x)}$ represents a trace preserving map from the output space to the subsequent input space (and therefore tracing the input space gives an identity operator on the previous output space). In the final line, we identify the $\mathtt{MM}$ process on the remaining times $\mathsf{C}^{\textup{MM}}_{N:1 \setminus j+1^{\out}j^{\inp}} := \sum_x p_x \left(\bigotimes_{i=j+1}^{N-1} L_{i+1^{\inp}:i^{\out}}^{(x)} \right)\left( \bigotimes_{i=1}^{j-1} L_{i+1^{\inp}:i^{\out}}^{(x)} \otimes \rho_{1^{\inp}}^{(x)} \right)$. This implies that $\mathtt{MM} \subseteq \mathtt{NS}$. Furthermore, one can see that $\mathtt{MM}$ form a strict subset of $\mathtt{NS}$ since the former are defined as the convex hull of memoryless quantum processes, whereas the latter are built by affine combinations; crucially, there exist such affine (but not convex) combinations that lead to valid processes. This latter point can be seen by considering Lem.~\ref{lem::nseqqm-2}, which states that for $N=2$ times, $\mathtt{NS} = \mathtt{QM}$; we also clearly have that $\mathtt{MM} \subsetneq \mathtt{NS}$ (which can be seen by considering any process that is entangled across the partition $1^{\inp} : 1^{\out} 2^{\inp}$ [see Ex.~\ref{ex::entangledns}]). Since the former coincide with the set of affine combinations of $\mathtt{M}$ whereas the latter coincide with the set of convex combinations, these sets of processes differ and their exclusion is non-empty.
\end{proof}

As a side remark, note that this lemma implies that the set of $\mathtt{MM}$ processes are a measure-zero subset of quantum processes, as they are confined to a manifold of strictly lower dimension than the full space of quantum processes, as evidenced by the linear non-signalling constraints on the process; so true is the case for $\mathtt{NS}$. This fact notwithstanding, such processes provide an interesting case study of practical relevance, as the former provide perhaps the simplest non-trivial quantum processes to implement, while the latter are crucial for uncovering causal relations between events.

\begin{proof}[Proof: Lem.~\ref{lem::cmbutsignalling}.]

In terms of operational consequences, the no-signalling condition here means that if one chooses any time $t_j$ at which they feed in a fresh state to the process and at the next timestep $t_{j+1}$ they discard the state output by the process, then \emph{all} observed phenomena on the remaining times are independent of any state input at time $t_j$. Although this holds true for mixed memoryless quantum processes, it is not the case for the more general classical memory quantum process, which can exhibit (classical) signalling from any point in time to any other in the future via the classical memory mechanism made manifest through the environment. We now present an example that demonstrates this, thereby proving our claim.

The process is as follows (see Fig.~\ref{fig::circuit}): an arbitrary state of the system $\rho$ is accessible to the agent at time $t_1$. Following the agent's (arbitrary) interrogation, the post-measurement system state is swapped with that of the environment, which begins in the fiducial state $\ket{0}$; the agent can then measure the system state at time $t_2$. Concurrently, the environment is subject to the EBC which measures it in the computational basis. Subsequently, a $\mathtt{CNOT}$ gate is applied to the system (controlled on the environment), which can finally be interrogated at time $t_3$. 

Although this process is a $\mathtt{CM}$ one (by construction), it exhibits signalling from time $t_1$ to $t_3$, which is in contradiction with the structure of a $\mathtt{MM}$ quantum process (by Lem.~\ref{lem::MMnosignalling}). More precisely, we have that $\ptr{2^{\inp}}{\mathsf{C}_{3:1}} \neq \mathsf{C}_{3^{\inp} 2^{\out} 1^{\inp}} \otimes \mathbbm{1}_{1^{\out}}$, which can be seen as follows. First, note that the non-signalling structure on the r.h.s. here implies that any (trace-preserving) reduction to a process from the input space on time $t_1$ to that on time $t_3$ must be of the form $\rho_{3^{\inp}} \otimes \mathbbm{1}_{1^{\out}}$ for some quantum state $\rho$. Then consider tracing the input spaces associated to times $t_1$ and $t_2$ (corresponding to the act of discarding the system states), and feeding in an arbitrary but fixed state $\sigma_{2^{\out}}$ at time $t_2$. In doing so, one yields a reduced channel from the output at time $t_1$ to the input at time $t_3$, which 
\begin{align}
    \mathsf{C}_{3^{\inp} 1^{\out}}(\sigma_{2^{\out}}) := \mathsf{C}_{3:1} \star \mathbbm{1}_{2^{\inp}} \star \mathbbm{1}_{1^{\inp}} \star \sigma_{2^{\out}} \neq \rho_{3^{\inp}} \otimes \mathbbm{1}_{1^{\out}}.
\end{align}
The non-equivalence can be demonstrated by noting that a process of the form $\rho_{3^{\inp}} \otimes \mathbbm{1}_{1^{\out}}$ takes \emph{any} input state $\tau_{1^{\out}}$ at time $t_1$ to the unique state $\rho_{3^{\inp}}$; however, it is straightforward to calculate that \mbox{$\rho_{3^{\inp}}(\tau_{1^{\out}} = \ket{0}; \sigma_{2^{\out}}) := \mathsf{C}_{3^{\inp} 1^{\out}}(\sigma_{2^{\out}}) \star \ketbra{0}{0}_{1^{\out}} = \sigma_{3^{\inp}}$}, whereas $\rho_{3^{\inp}}(\tau_{1^{\out}}^\prime = \ket{1}; \sigma_{2^{\out}}) := \mathsf{C}_{3^{\inp} 1^{\out}}(\sigma_{2^{\out}}) \star \ketbra{1}{1}_{1^{\inp}} = \mathtt{NOT}(\sigma)_{3^{\inp}}$, evidencing signalling.

\end{proof}

\section{Convexity of \texorpdfstring{$\mathtt{CM}$}{}}\label{app::convexity}
\begin{thm}
    The set $\mathtt{CM}$ is convex.
\end{thm}

\begin{proof}
One way to prove that $\mathtt{CM}$ is convex is to notice the that the characterisation of classical memory processes presented in Eq.~\eqref{eq::classicalmemoryquantumprocess-2} is a convex combination. More precisely, as proved in App.~\ref{app::classicalmemoryquantumprocesses},  Eq.~\eqref{eq::classicalmemoryquantumprocess-2} states that an $N$-time quantum process belongs to $\mathtt{CM}$ if and only if
\begin{align}
    \mathsf{C}_{N:1}^{\textup{CM}}\!=\!  \sum_{x_{N:1}} \mathsf{L}_{N^{\inp}:N-1^{\out}}^{(x_N|x_{N-1:1})} \otimes \hdots \otimes \mathsf{L}_{2^{\inp}:1^{\out}}^{(x_2|x_1)} \otimes \rho_{1^{\inp}}^{(x_1)},
\end{align}
where $x_{k:j} := \{ x_j,\hdots, x_k\}$, $\{ \rho_{1^{\inp}}^{(x_1)} \}$ forms a state ensemble, i.e., each $\rho_{1^{\inp}}^{(x_1)} \geq 0$ and $\rho_{1^{\inp}} := \sum_{x_1} \rho_{1^{\inp}}^{(x_1)}$ has unit trace, and $\mathsf{L}_{j+1^{\inp}:j^{\out}}^{(x_{j+1}|x_{j:1})} := \mathsf{M}_{E_{j+1^{\inp}}}^{(x_{j+1})} \star \mathsf{U}_{(ES)_{j+1^{\inp} j^{\out}}} \star \sigma_{E_{j^{\out}}}^{(x_{j:1})}$ forms an instrument for each conditioning argument. 

To make the statement more explicit, we now show that for every probability $q\in [0,1]$ and every classical memory processes $ \mathsf{C}_{N:1}^{A}, \mathsf{C}_{N:1}^{B}\in \mathtt{CM}$, we have
\begin{align}
   \mathsf{C}_{N:1}:= q \mathsf{C}_{N:1}^{A} + (1-q)\mathsf{C}_{N:1}^{B} \in \mathtt{CM}.
\end{align}
By Eq.~\eqref{eq::classicalmemoryquantumprocess-2}, we can write $\mathsf{C}_{N:1}^{A}$ and $\mathsf{C}_{N:1}^{B}$ as 
\begin{align}
    \mathsf{C}_{N:1}^{A}\!=\! & \sum_{x^A_{N:1}} \mathsf{L}_{N^{\inp}:N-1^{\out}}^{(x^A_N|x^A_{N-1:1})} \otimes \hdots \otimes \mathsf{L}_{2^{\inp}:1^{\out}}^{(x^A_2|x^A_1)} \otimes \rho_{1^{\inp}}^{(x^A_1)} \notag \\
    \mathsf{C}_{N:1}^{B}\!=\! & \sum_{x^B_{N:1}} \mathsf{L}_{N^{\inp}:N-1^{\out}}^{(x^B_N|x^B_{N-1:1})} \otimes \hdots \otimes \mathsf{L}_{2^{\inp}:1^{\out}}^{(x^B_2|x^B_1)} \otimes \rho_{1^{\inp}}^{(x^B_1)}.
\end{align}
We now define the unnormalised states $\hat{\rho}_{1^{\inp}}^{(x^A_1)}:=q \rho_{1^{\inp}}^{(x^A_1)}$ and $\hat{\rho}_{1^{\inp}}^{(x^B_1)}:=(1-q)\rho_{1^{\inp}}^{(x^B_1)}$, the unnormalised processes
\begin{align}
    \widehat{\mathsf{C}}_{N:1}^{A}\!=\! & \sum_{x^A_{N:1}} \mathsf{L}_{N^{\inp}:N-1^{\out}}^{(x^A_N|x^A_{N-1:1})} \otimes \hdots \otimes \mathsf{L}_{2^{\inp}:1^{\out}}^{(x^A_2|x^A_1)} \otimes \hat{\rho}_{1^{\inp}}^{(x^A_1)} \notag \\
    \widehat{\mathsf{C}}_{N:1}^{B}\!=\! & \sum_{x^B_{N:1}} \mathsf{L}_{N^{\inp}:N-1^{\out}}^{(x^B_N|x^B_{N-1:1})} \otimes \hdots \otimes \mathsf{L}_{2^{\inp}:1^{\out}}^{(x^B_2|x^B_1)} \otimes \hat{\rho}_{1^{\inp}}^{(x^B_1)},
\end{align}
and a set which contains all indices from $\mathsf{C}_{N:1}^{A}$ and $\mathsf{C}_{N:1}^{B}$, $\mathcal{S}:=\{x^A_{N:1}\}\cup\{x^B_{N:1}\} $. This allows us to write
\begin{align}
    \mathsf{C}_{N:1} &= q \mathsf{C}_{N:1}^{A} + (1-q)\mathsf{C}_{N:1}^{B} \notag \\
    &= \widehat{\mathsf{C}}_{N:1}^{A} + \widehat{\mathsf{C}}_{N:1}^{B} \notag \\
    &= \sum_{x_{N:1}\in \mathcal{S}} \mathsf{L}_{N^{\inp}:N-1^{\out}}^{(x_N|x_{N-1:1})} \otimes \hdots \otimes \mathsf{L}_{2^{\inp}:1^{\out}}^{(x_2|x_1)} \otimes \rho_{1^{\inp}}^{(x_1)} 
\end{align}
which is a classical memory quantum process.
\end{proof}

\end{document}